\documentclass[letterpaper, oneside, 12pt]{article}
\usepackage{amsmath}
\usepackage{graphicx}%
\usepackage{amsfonts}%
\usepackage{amssymb}
\usepackage{xcolor}
\usepackage{overpic}
\usepackage[ruled]{algorithm}
\usepackage[noend]{algpseudocode}
\usepackage{subfigure}
\usepackage{rotating}
\usepackage{epstopdf}
\usepackage{url}
\usepackage{comment}
\usepackage{enumerate}
\usepackage{mathtools}
\usepackage{multirow}
\usepackage{pgfplots}
\usepackage{natbib}
\usepackage{pst-solides3d}

\usepackage{setspace}
\usepackage[hmargin=1.0in,vmargin=1.0in]{geometry}

\usepackage[compact]{titlesec}
\titlespacing{\section}{0pt}{*0.8}{*0.8}
\titlespacing{\subsection}{0pt}{*0.8}{*0.8}
\titlespacing{\subsubsection}{0pt}{*0.8}{*0.8}

\newcommand{\bA}{ {\boldsymbol A} }

\newcommand{\bC}{ {\boldsymbol C} }

\newcommand{\bD}{ {\boldsymbol D} }

\newcommand{\bI}{ {\boldsymbol I} }

\newcommand{\bJ}{ {\boldsymbol J} }

\newcommand{\bv}{ {\boldsymbol v} }

\newcommand{\bx}{ {\boldsymbol x} }
\newcommand{\bX}{ {\boldsymbol X} }
\newcommand{\by}{ {\boldsymbol y} }

\newcommand{\bbeta}{ {\boldsymbol \beta} }

\newcommand{\bDelta}{ {\boldsymbol \Delta} }

\newcommand{\bepsilon}{ {\boldsymbol \epsilon} }

\newcommand{\bPhi}{ {\boldsymbol \Phi} }

\newcommand{\blambda}{ {\boldsymbol \lambda} }

\newcommand{\bSigma}{ {\boldsymbol \Sigma} }

\newcommand{\bzeta}{ {\boldsymbol \zeta} }

\newcommand{\bxi}{ {\boldsymbol \xi} }

\newcommand{\bzero}{ {\boldsymbol 0} }

\newtheorem{theorem}{Theorem}[section]
\newtheorem{lemma}[theorem]{Lemma}

\newenvironment{proof}[1][Proof]{\begin{trivlist}
\item[\hskip \labelsep {\bfseries #1}]}{\end{trivlist}}

\newcommand{\qed}{\nobreak \ifvmode \relax \else
      \ifdim\lastskip<1.5em \hskip-\lastskip
      \hskip1.5em plus0em minus0.5em \fi \nobreak
      \vrule height0.75em width0.5em depth0.25em\fi}
\title{Sketching in Bayesian High Dimensional Regression With Big Data Using Gaussian Scale Mixture Priors}
\author{{\small Rajarshi Guhaniyogi}\\
{\small Associate Professor, Department of Statistics,}\\ {\small UC Santa Cruz, 1156 High Street, Santa Cruz, CA 95064, E-mail: rguhaniy@ucsc.edu}\\
{\small Aaron Scheffler}\\
{\small Assistant Professor, Department of Epidemiology \& Biostatistics,}\\
{\small UC San Francisco, 550 16th. Street
San Francisco CA 94158, E-mail: Aaron.Scheffler@ucsf.edu }}
\begin{document}
\maketitle
\begin{abstract}
Bayesian computation of high dimensional linear regression models with a popular Gaussian scale mixture prior distribution using Markov Chain Monte Carlo (MCMC) or its variants can be extremely slow or completely prohibitive due to the heavy computational cost that grows in the order of $p^3$, with $p$ as the number of features. Although a few recently developed algorithms make the computation efficient in presence of a small to moderately large sample size (with the complexity growing in the order of $n^3$), the computation becomes intractable when sample size $n$ is also large. In this article we adopt the data sketching approach to compress the $n$ original samples by a random linear transformation to $m<<n$ samples in $p$ dimensions, and compute Bayesian regression with Gaussian scale mixture prior distributions with the randomly compressed response vector and feature matrix. Our proposed approach yields computational complexity growing in the cubic order of $m$. Another important motivation for this compression procedure is that it anonymizes the data by revealing little information about the original data in the course of analysis. Our detailed empirical investigation with the Horseshoe prior from the class of Gaussian scale mixture priors shows closely similar inference and a massive reduction in per iteration computation time of the proposed approach compared to the regression with the full sample. One notable contribution of this article is to derive posterior contraction rate for high dimensional predictor coefficient with a general class of shrinkage priors on them under data compression/sketching. In particular, we characterize the dimension of the compressed response vector $m$ as a function of the sample size, number of predictors and sparsity in the regression to guarantee accurate estimation of predictor coefficients asymptotically, even after data compression.
\end{abstract}
\section{Introduction}
Of late, due to the technological advances in a variety of disciplines, we routinely encounter data with a large number of predictors. In such settings, it is commonly of interest to consider the high dimensional linear regression model
\begin{align}\label{eq:GLM}
y=\bx'\bbeta+\epsilon,
\end{align}
where $\bx$ is a $p\times 1$ feature vector, $\bbeta$ is the corresponding $p\times 1$ coefficient, $y$ is the continuous response and $\epsilon$ is the idiosyncratic error. Bayesian methods for estimating $\bbeta$ broadly employ two classes of prior distributions. The traditional approach is to develop a discrete mixture of prior
distributions \citep{george1997approaches,scott2010bayes}. These methods enjoy the advantage
of inducing exact sparsity for a subset of parameters (allowing some components of $\bbeta$ to be exactly zero a posteriori) and minimax rate of posterior contraction \citep{castillo2015bayesian} in high dimensional regression, but face computational challenges when the number of features is even moderately large. As an alternative to this approach, continuous shrinkage priors \citep{armagan2013generalized,carvalho2010horseshoe,caron2008sparse} have emerged, which can mostly be expressed as global-local scale mixtures of Gaussians \citep{polson2010shrink} given by,
\begin{align}\label{Gauss_scale}
&\beta_j|\lambda_j,\tau,\sigma\sim N(0,\sigma^2\tau^2\lambda_j^2),\:\lambda_j\sim g_1,\:\mbox{for}\: j=1,...,p\nonumber\\
&\qquad\qquad\qquad\qquad\tau\sim g_2,\:\sigma\sim f,
\end{align}
where $\tau$ is known as the local parameter and $\lambda_j$'s are known as the global parameters, $g_1,g_2$ and $f$ are densities supported on $\mathbb{R}^{+}$. The prior structure (\ref{Gauss_scale}) induces approximate sparsity
in $\beta_j$ by shrinking the null components toward zero while retaining the true
signals \citep{polson2010shrink}. The global parameter $\tau$ controls the number of
signals, while the local parameters $\lambda_j$ dictate whether they are nulls. In this sense, the prior
(\ref{Gauss_scale}) approximates the properties of point-mass mixture priors \citep{george1997approaches,scott2010bayes}.
% Different choices of $g_1$ and $g_2$ lead to different classes of Bayesian shrinkage priors. For example, the state-of-the-art Horseshoe shrinkage prior (\citealp{carvalho2010horseshoe}) is obtained by choosing $g_1$ and $g_2$ both as half Cauchy distributions.

%and offer an approximation to the operating characteristics of discrete mixture priors.
Global-local priors allow parameters to be updated in blocks via a fairly automatic Gibbs sampler that leads to rapid mixing and convergence of the resulting Markov chain.
In particular, letting $\bX$ be the $n\times p$ feature matrix, $\by$ be the $n\times 1$ response vector and $\bDelta=\tau^2diag(\lambda_1,...,\lambda_p)$, the distribution of $\bbeta=(\beta_1,...,\beta_p)'$ conditional on $\blambda=(\lambda_1,...,\lambda_p)',\tau$, $\sigma$, $\by$ and $\bX$ follows $N((\bX'\bX+\bDelta^{-1})^{-1}\bX'\by,\sigma^2(\bX'\bX+\bDelta^{-1})^{-1})$, and
can be updated in a block. On the other hand, $\lambda_j$'s are conditionally independent and allow fairly straightforward updating using either Gibbs sampling or slice sampling.
The posterior draws from $\bbeta,\blambda,\tau,\sigma$ are found to offer an accurate approximation to the operating characteristics of discrete mixture priors.
However, sampling from the full conditional posterior of $\bbeta$ require storing and computing the Cholesky decomposition of the $p\times p$ matrix $(\bX'\bX+\bDelta^{-1})$, that necessitates $p^3$ floating point operations (flops) and $p^2$ storage units, which can be severely prohibitive for large $p$.
%There are recent advances in fast sampling from the conditional distribution of $\beta$ by exploiting linear algebra artifacts in regressions
Recent work in high dimensional regressions involving small $n$ and large $p$ \citep{bhattacharya2016fast} exploits the Woodbury matrix identity to draw from the full conditional posterior distribution of $\bbeta$ by inverting only an $n\times n$ matrix. When $n$ is large, this algorithm is embedded within an approximate MCMC sampling framework \cite{johndrow2020scalable} to facilitate fast computation. 
%In this article, we will propose an alternative solution for big $n$ and $p$ using the data compression idea.
%There are recent advances in fast sampling of Gaussian scale mixture priors for large $p$ and small $n$ (Bhattacharya et al., 2015), though Hager, W. W. (1989). Updating the inverse of a matrix. SIAM review 31, 221–239.

Following the literature on \emph{Sketching}, we propose to compress the response vector and feature matrix by a random linear transformation,
reducing the number of records from $n$ to $m$, while preserving the number of original features. The compressed version of the original dataset,
referred to as a \emph{sketch}, then serves as a surrogate for a high dimensional regression analysis with a suitable Gaussian scale mixture prior on the
feature coefficients. Since the number of compressed records $m$ is much smaller than the sample size $n$, one can adapt existing algorithms on the compressed data for efficient estimation of posterior distribution for feature coefficients with large number of features and large sample. On the theoretical front, we assume that the shrinkage priors of our interest have densities with a dominating peak around $0$ and flat, heavy tails, and have sufficient mass around the true regression coefficient. We then identify conditions on the predictor matrix, the interlink between the dimension of the random compression matrix, sample size, sparsity of the true regression coefficient vector and the number of features to prove optimal convergence rate of estimating the predictor coefficients asymptotically under data compression. Our empirical investigation ensures that the relevant features can be accurately learnt from the compressed data. Moreover, in presence of a higher degree of sparsity in the true regression model, the actual estimates of parameters and predictions are as accurate as they would have been, had the uncompressed data been used. Another attractive feature of this approach is that the original data are not recoverable from the compressed data, and the compressed data effectively reveal no more information than would be revealed by a completely new sample. In fact, the original uncompressed data does not need to be stored in the course of the analysis. While the core idea behind the development apply broadly to the class of global-local priors (\ref{Gauss_scale}), for sake of concreteness our detailed empirical investigation focuses on the popular horseshoe prior \citep{carvalho2010horseshoe} which corresponds to both $g_1$ and $g_2$ in (\ref{Gauss_scale}) being the half-Cauchy distribution. The horseshoe achieves the minimax adaptive rate of contraction when the true $\bbeta$ is sparse \citep{van2014horseshoe,van2017adaptive} and is considered to be among the state-of-the-art shrinkage priors.

%\textcolor{blue}{Write the novelty about the real data.}

%Moreover, the marginal
%credible intervals have asymptotically correct frequentist coverage van der Pas et al. (2017b) for parameters that are either very close to zero or above the detection threshold, though signals in a certain “intermediate” range are shrunk too much toward zero for credible intervals to have correct coverage.

%\textcolor{blue}{Will make this paragraph smooth once we come to the final stage.} 
In this context, it is worth mentioning the contribution of this article in light of the relevant literature of sketching in regression. Sketching has become an increasingly popular research topic in the machine learning literature in the last decade or so, see \cite{vempala2005random,halko2011finding,mahoney2011randomized,woodruff2014sketching} and references therein. In the context of high dimensional linear regressions, sketching has been employed to study various aspects of ridge regression, referred to as the sketched ridge regression. \cite{zhang2013recovering} study the dual problem in a complementary finite-sample setting, where as \cite{chen2015fast} propose an algorithm combining sparse embedding and the subsampled randomized Hadamard transform (SRHT), proving relative approximation bounds. \cite{wang2017sketched} study iterative sketching algorithms from an optimization point of view, for both the primal and the dual problems. \cite{zhou2008compressed} show that identifying the correct sparse set of relevant variables by the lasso are as effective under data sketching.  \cite{dobriban2018new} study sketching using asymptotic random matrix theory, but only for un-regularized linear regression. \cite{chowdhury2018iterative} propose a data-dependent algorithm in light of the ridge leverage scores. Other related works include \cite{ailon2006approximate,drineas2011faster, raskutti2016statistical,ahfock2017statistical,huang2018near}. To the best of our knowledge, we are the first to offer efficient and principled Bayesian computation algorithm with linear regressions involving large $n$ and $p$ using data sketching. Moreover, to the best of our knowledge, the theoretical result on the posterior convergence rate of regression parameters under data compression has not been established before.

Our proposal is related to compressed sensing approaches \citep{donoho2006compressed,candes2006near,eldar2012compressed}, with an important difference. While compressed sensing approaches broadly aim at reconstructing a sparse $\bX$ from a small number of its random linear combinations, we intend to
reconstruct a sparse function of $\bX$ only, and not the $\bX$ and $\by$ themselves. In fact, from our point of view of preserving privacy of the response vector and feature matrices, approximately reconstructing them should be viewed as undesirable. Our approach is fundamentally different from \cite{maillard2009compressed,guhaniyogi2015bayesian,guhaniyogi2016compressed} in that they compress each feature vector, leading to an $m$-dimensional compressed features from  $p$-dimensional features for each sample. In contrast, our compression framework does not alter the number of features in the analysis before and after compression.
%A few notable articles in the machine learning literature show that major statistical procedures, such as the principal component analysis, clustering, and even identifying the correct sparse set of relevant variables by the lasso are as effective under compression \citep{liu2005random,zhou2008compressed}. However, we are not aware of any earlier work where the full potential of the data compression approach to enable efficient Bayesian computation with big $n$ and $p$ has been carefully studied both theoretically and empirically.

The rest of the article proceeds as follows. Section~\ref{sec2} details out the proposed model and algorithm for efficient estimation of feature coefficients
in presence of large $n$ and $p$. Section~\ref{sec3} offers theoretical insights into the choice of $m$ as a function of the true sparsity, number of features and sample size $n$ to obtain accurate estimation of feature coefficients asymptotically. Section~\ref{sec4} empirically investigates parametric and predictive inferences from the proposed approach with the horseshoe shrinkage prior under various simulation cases. The proposed method is illustrated on a real data with big $p$ and $n$ in Section~\ref{sec5}, followed by the concluding remarks
in Section~\ref{sec6}.

\section{Sketching Response Vector and Feature Matrix for Large $n$}\label{sec2}
For subjects $i=1,...,n$, let $y_i\in\mathcal{Y}$ denote the response for subject $i$ corresponding to the feature $\bx_i\in\mathcal{R}^p$. This article focuses on the scenario where $n$ and $p$ both large. Let $\by=(y_1,...,y_n)'$ be the $n\times 1$ vector of responses and $\bX=[\bx_1:\cdots:\bx_p]'$ be the $n\times p$ matrix of predictors. As a first step to our proposal, we consider a sketching or data compression approach by pre-multiplying $\by$ and $\bX$ with a sketching matrix $\bPhi$ of dimension $m\times n$ with $m<<n$ to construct data sketches $\tilde{\by}=\bPhi\by$ and $\tilde{\bX}=\bPhi\bX$ of dimensions $m\times 1$ and $m\times p$, respectively. The data sketches are employed
to set up the high dimensional linear regression having the form
\begin{align}\label{model1}
\tilde{\by}=\tilde{\bX}\bbeta+\bepsilon,\:\bepsilon\sim N(0,\sigma^2 \bI),
\end{align}
where $\sigma^2$ is the idiosyncratic error variance. We do not estimate $\bPhi$ as a variable in the regression, rather follow the idea of data oblivious sketches
to construct $\bPhi$ prior to fitting the model (\ref{model1}). More specifically, following the idea of Gaussian sketching \citep{sarlos2006improved}, the elements $\Phi_{ij}$ of the $\bPhi$ matrix are drawn independently from N($0,1/n$). The computational complexity of obtaining the sketched data using Gaussian sketches is given by $O(mnp)$. While there are more computationally efficient data oblivious options for random projection/sketching matrix $\bPhi$, such as the Hadamard sketch \citep{ailon2009fast} and the Clarkson-Woodruff sketch \citep{clarkson2017low}, we find it to be less concerning in our framework since the computation time for Bayesian fitting of (\ref{model1}) far exceeds the difference in time for computing sketched data with different options of sketching matrices.

The data compression approach implemented here appears to be a special case of the \emph{matrix masking} technique proposed in the earlier privacy literature \citep{ting2008random,zhou2008compressed,zhao2019privacy}, which, although popular in the privacy literature, has not been given due attention theoretically, especially from a Bayesian perspective. A typical matrix masking procedure pre- and post-multiplies the data matrix $\bX$ by matrices $\bC$ and $\bD$, respectively, and releases $\bC\bX\bD$ for the ensuing analysis. The transformation is quite general, and allows the possibility of deleting records, suppressing subsets of variables and data swapping.
This article chooses $\bC=\bPhi$ and $\bD$ as the identity matrix so as to keep the original interpretation of the features. Notably, even in the case of $\bPhi$ being known, the linear system
$\bPhi \bX$ is grossly under-determined due to $m<<min(n,p)$. The privacy in information theoretic terms of this sketching procedure could be evaluated using an upper bound of the average mutual information $\mathcal{I}(\tilde{\bX}, \bX)/np$ per unit in the original data matrix $\bX$, and showing that $Sup\:\mathcal{I}(\tilde{\bX}, \bX)/np = O(m/n)$ \citep{zhou2008compressed}, where supremum is taken over all possible distributions of $\bX$. With $m$ growing at a much slower rate than $n$, asymptotically as $n\rightarrow\infty$, the supremum over average mutual information converges to $0$, intuitively meaning that the compressed data reveal no more information about the original data than could be obtained from an independent sample.
It is be noted that such a bound is obtained assuming that $\bPhi$ is known. In practice, only $\tilde{\bX}=\bPhi \bX$ (and not even $\bPhi$) will be revealed to the analyst. Hence, the imposed privacy through compression is more strict than what is revealed by this result.

Although not apparent, the ordinary high dimensional regression model in (\ref{eq:GLM}) bears a close connection with its computationally convenient alternative (\ref{model1}), especially for large $n$. To see this, note that pre-multiplying the high dimensional linear regression equation $\by=\bX\bbeta+\bepsilon$ by $\bPhi$ results in
\begin{align}\label{comp_eq}
\bPhi \by=\bPhi \bX\bbeta+\tilde{\bepsilon},\:\:\tilde{\bepsilon}\sim N(\bzero,\sigma^2\bPhi\bPhi').
\end{align}
Equations (\ref{comp_eq}) and (\ref{model1}) are similar in the mean function but differ in the error distribution. More specifically,
our approach assumes components of the error vector $\bepsilon$ are i.i.d., whereas the error vector from (\ref{comp_eq}) follows a $N(\bzero,\sigma^2\bPhi\bPhi')$ distribution.
 Lemma 5.36 and Remark 5.40 of \cite{vershynin2010introduction} show that $||\bPhi\bPhi'-\bI_m||_2\leq C'\sqrt{m/n}$, with probability at least $1-e^{-C''m}$, for some constants $C',C''>0$.
As $m$ grows at a slower rate than $n$, $m/n\rightarrow 0$ asymptotically. Hence, with large $n$, the error distributions of (\ref{model1}) and (\ref{comp_eq}) behave similarly with a probability close to $1$.

With prior distribution on $\bbeta$ set as a Gaussian scale-mixture distribution from the class of distributions given by (\ref{Gauss_scale}), posterior computation using a blocked Metropolis-within-Gibbs algorithm cycles through updating the full conditional distributions: (a) $\bbeta|\blambda,\sigma,\tau$, (b) $\blambda|\bbeta,\sigma,\tau$, (c) $\sigma|\blambda,\bbeta,\tau$ and (d) $\tau|\blambda,\bbeta,\sigma$. Explicit expressions for (a), (b), (c) and (d) for the horseshoe shrinkage priors \citep{carvalho2010horseshoe} are available in Appendix A. While updating (b), (c) and (d) do not face any computational challenge due to big $n$ or $p$, full conditional posterior updating of $\bbeta|\blambda,\sigma,\tau$ has the form given by
\begin{align}\label{cond_beta}
N\left(\left(\tilde{\bX}'\tilde{\bX}+\bDelta^{-1}\right)^{-1}\tilde{\bX}'\tilde{\by},\sigma^2(\tilde{\bX}' \tilde{\bX}+\bDelta^{-1})^{-1}\right),\:\:\bDelta=\tau^2\mbox{diag}(\lambda_1,...,\lambda_p).
\end{align}
The most efficient algorithm to sample from $\bbeta$ \citep{rue2001fast} computes Cholesky decomposition of $\left(\tilde{\bX}'\tilde{\bX}+\bDelta^{-1}\right)$ and employs the Cholesky factor to solve a series of linear systems to draw a sample from (\ref{cond_beta}). In absence of
any easily exploitable structure, computing and storing the Cholesky factor of this matrix involves $O(p^3)$ and $O(p^2$) floating point operations respectively \citep{golub2012matrix}, which leads to computational and storage bottlenecks with a large $p$. To overcome the computational and storage burden, we adapt the recent algorithm proposed in the context of uncompressed data with small sample size \citep{bhattacharya2016fast} to our setting. The detailed steps are given as following:
\begin{description}
\item \textbf{Step 1:} Draw $\bv_1\sim N(\bzero,\sigma^2\bDelta)$ and $\bv_2\sim N(\bzero,\bI_m)$
\item \textbf{Step 2:} Set $\bv_3=\tilde{\bX} \bv_1/\sigma+\bv_2$.
\item \textbf{Step 3:} Solve $(\tilde{\bX}\bDelta \tilde{\bX}'+\bI_m)\bv_4=(\tilde{\by}/\sigma-\bv_3)$.
\item \textbf{Step 4:} Set $\bv_5=\bv_1+\sigma\bDelta \tilde{\bX}'\bv_4$.
\end{description}
$\bv_5$ is a draw from the full conditional posterior distribution of $\bbeta$. Notably, the computational complexity of Steps 1-4 is dominated by two operations:
(Operation A) computing the inverse of $(\bPhi \bX\bDelta \bX'\bPhi'+\bI_m)$, and (Operation B) calculating $\bPhi \bX\bDelta \bX'\bPhi'$. (Operation A) leads to a complexity of $O(m^3)$, whereas (Operation B) incurs complexity of
$O(m^2p$). As we demonstrate in Section 4, the algorithm offers massive speed-up in computation with big $p$ and $n$, since $m<<min(n,p)$. Notably, an application of \cite{bhattacharya2016fast} on the uncompressed data would have incurred computational complexity dominated by $O(n^3)$ and $O(n^2p)$. Thus, our compression approach helps speeding up computation in our empirical investigations with big $n$ and $p$.

One important question arises as to how much inference is lost in lieu of the computational speed-up achieved by the data compression approach. In the sequel, we address this question both theoretically and empirically. Section~\ref{sec3} derives theoretical conditions on $m$, $n$, $p$ and the sparsity of the true data generating model to show asymptotically desirable estimation of feature coefficients. Thereafter, finite sample performance of the proposed approach is presented both in the simulation study and in the real data section.

\section{Posterior Concentration Properties of the Sketching Approach}\label{sec3}
This section studies convergence properties of the data sketching approach with high dimensional shrinkage prior on predictor coefficients. In particular, we will establish the posterior contraction rate of estimating the predictor coefficient vector for the proposed model (\ref{model1}) under mild regularity conditions. To begin with, we define a few notations.
\subsection{Notations}
In what follows, we add a subscript $n$ to the dimension of the number of features $p_n$ and the dimension of the compression matrix $m_{n}$ to indicate that both of them increase with the sample size $n$. This asymptotic paradigm is also meant to capture the fact that the number of rows of the sketching matrix $m_n$ is smaller than the sample size $n$.
Naturally, the response vector $\by$, feature matrix $\bX$, feature coefficient vector $\bbeta$ and the sketching matrix $\bPhi$ are also functions of $n$. We denote them by $\by_n$, $\bX_n$, $\bbeta_n$ and $\bPhi_n$, respectively. Note that the true data generating model under data sketching is given by (\ref{comp_eq}). We use superscript $*$ to indicate the true parameters $\bbeta_n^*$ and $\sigma^{*2}$. For simplicity in the algebraic manipulation, we assume that $\sigma^2=\sigma^{*2}$ are both known and fixed at $1$. This is a common assumption in asymptotic studies \citep{vaart2011information}. Furthermore, it is known that the theoretical results obtained by assuming $\sigma^2$ as a fixed value is equivalent to those obtained by assigning a prior with a bounded support on $\sigma^2$ \citep{van2009adaptive}. $P_{\bbeta_n^*}$ denotes probability distribution under the true data generating model (\ref{comp_eq}).  For vectors, we let $||\cdot||_1, ||\cdot||_{2}$ and $||\cdot||_{\infty}$ denote the $L_1, L_2$ and $L_{\infty}$ norms, respectively. The number of nonzero elements in a vector is given by $||\cdot||_0$. The quantities $e_{min}(\bA)$ and $e_{max}(\bA)$ respectively represent the minimum and maximum eigenvalues of a square matrix $\bA$. We use $\{\theta_n\}$ to denote the Bayesian posterior contraction rate which satisfies $\theta_n\rightarrow 0$.

\subsection{Assumptions, Framework and The Main Result}
For any subset of indices $\bxi\subset\{1,...,p_n\}$, $|\bxi|$ denotes the number of elements in the index set $\bxi$. Depending on whether $\bA$ is a vector or a matrix, $\bA_{\bxi}$ denotes the sub-vector or the sub-matrix corresponding to the indices $\bxi$. We let $\bxi^*=\{j:\beta_{j,n}^*\neq 0$\}, i.e., $\bxi^*$ are the indices of the nonzero entries for the true predictor coefficient $\bbeta_n^*$, and $s_{n}$ (dependent on $n$) designates the number of nonzero entries in $\bbeta_n^*$, i.e., $s_{n}=||\bbeta_{n}^{*}||_{0}=|\bxi^*|$.
Since the shrinkage prior on $\bbeta_{n}$ assigns zero probability at the point zero, the exact number of nonzero elements of $\bbeta_{n}$ is always $p_n$.
%A meaningful comparison with the value $s_{n}$ is made by considering $\tilde{s}_{n}$, the number of elements of $\bbeta_{n}$ exceeding in absolute value a threshold $a_n$, which will be specified later. In other words, only elements with absolute value larger than $a_n$ will be treated as significant and counted towards non-zero entries.
Before rigorously studying properties of the posterior distribution, we state some regularity conditions on the design matrix $\bX_n$, the compression matrix $\bPhi_n$ and the true sparsity $s_n$.
%and the model determined approximate sparsity $\tilde{s}_n$.\\
\begin{description}
\item (A) All covariates are uniformly bounded, let $|x_{i,j}|\leq 1$, for all $i=1,...,n$ and $j=1,..,p_n$.
\item (B) $||\bPhi_n\bPhi_n'-\bI_{m_n}||_2\leq C'\sqrt{m_n/n}$, for some constant $C'>0$, for all large $n$.
\item (C) $s_n\log(p_n)=o(m_n)$, $m_n=o(n)$.
%\item (D) $m_n=o(n)$ and $n^{1/2+\tilde{\delta}}/m_n\rightarrow 0$ for some $\tilde{\delta}>0$.
%\item (E) $\tilde{s}_n=O(s_n)$.
\item (D) $e_{min}(\bX_{n,\bxi}'\bX_{n,\bxi}/n)\geq \eta$, for some $\eta>0$ and for all $\bxi\supset\bxi^*$ such that $|\bxi|\leq \bar{s}_n$, where $\bar{s}_n$ satisfies $s_n=o(\bar{s}_n)$.
\end{description}
(A) is a common assumption in the context of compressed sensing, see \cite{zhou2008compressed}. From the theory of random matrices, (B) occurs with probability at least $1-e^{-C''m_n}$ (see Lemma 5.36 and Remark 5.40 of \cite{vershynin2010introduction}). Hence (B) is a mild assumption for large $n$. (C) restricts the growth of the true sparsity and presents an interlink between the true sparsity, the rank of the random matrix, number of predictor coefficients and the sample size. (D) puts restriction on the smallest eigenvalue of the matrix $\tilde{\bX}_{n,\bxi}'\tilde{\bX}_{n,\bxi}/m_n$. Notably, Gaussian sketching approximately preserves the isometry condition \citep{ahfock2017statistical}, so that $\exists$ $\eta_0>0$ with the property that $e_{min}(\tilde{\bX}_{n,\bxi}'\tilde{\bX}_{n,\bxi}/m_n)\geq \eta_0e_{min}(\bX_{n,\bxi}'\bX_{n,\bxi}/n)$ with probability $f_n$ depending on $m_n$ and $p_n$ . This, together with assumption A$_1$(3) in \cite{song2017nearly} is used to argue that assumption (D) is satisfied with a positive probability.

%(D) allows $m_n$ to grow at a slower rate than $n$, while at the same time ensuring a faster growth than $\sqrt{n}$ for $m_n$. Assumptions (A) and (C) jointly impose restrictions on the compressed predictor matrix $\tilde{X}=\Phi_n X$.  Following the proof of Proposition 3.6 in \cite{zhou2008compressed}, we observe that assumptions (A) and (C) imply
%$e_{min}(\tilde{X}_{\xi}'\tilde{X}_{\xi}/m_n)\geq \eta$, for some $\eta>0$ and for all $\xi\supset\xi^*$ such that $|\xi|\leq s_n+\tilde{s}_n$. We will make use of this fact
%in the proof of our posterior consistency result.
Our next set of assumptions concern the tail behavior of the shrinkage priors of interest and the magnitude of the nonzero entries of the true coefficient $\bbeta_n^*$. Let $h_{\mu_n}(x)$ denote the prior density of $\beta_{j,n}$ for all $j$ with the set of hyper-parameters $\mu_n$. For $a_n=\sqrt{s_n\log(p_n)/m_n}/p_n$ and for a sequence $M_n$ nondecreasing as a function of $n$, we assume
\begin{description}
\item (E) $\max\limits_{j\in\bxi^*}|\beta_{j,n}^*|<M_n/2$.
\item (F) $1-\int_{-a_n}^{a_n}h_{\mu_n}(x)dx\leq p_n^{-(1+u)}$, for some positive constant $u$.
\item (G) $-\log(\inf\limits_{x\in[-M_n,M_n]}h_{\mu_n}(x))=O(\log(p_n))$.
\end{description}
Assumption (E) restricts the growth of the nonzero entries in the true regression parameter asymptotically.
Assumption (F) concerns the prior concentration, requiring that the prior density of $\beta_{j,n}$ for all $j$ has sufficient mass within the interval $[-a_n,a_n]$.
Finally, Assumption (G) essentially controls the prior density around the true feature coefficient. Notably, Assumptions (E)-(G) are frequently used in the high dimensional
Bayesian regression literature, including in \cite{jiang2007bayesian} and \cite{song2017nearly}.

Define %$\mathcal{B}_n=\left\{\mbox{At least $\tilde{s}_{n}$ absolute values of $\bbeta_{n}$ are greater than $a_n=\sqrt{s_n\log(p_n)/m_n}/p_n$}\right\}$,\\
$\mathcal{A}_n=\left\{\bbeta_{n}:||\bbeta_{n}-\bbeta_{n}^{*}||_2>3\theta_n\right\}$, %and $\mathcal{A}_n=\mathcal{B}_n\cup\mathcal{C}_n$.
$\mathcal{B}_n=\{\mbox{At least}\:\tilde{s}_n\: \mbox{number of} \:|\beta_{k,n}|\geq a_n\}$, with $\tilde{s}_n=O(s_n)$, $\mathcal{C}_n=\mathcal{A}_n\cup\mathcal{B}_n$. Since the shrinkage prior assigns zero probability at point zero, the number of nonzero elements of $\bbeta_n$ is $p_n$. Thus the number of nonzero components of $\bbeta_n$ is assessed by considering the number of $\beta_{k,n}$'s which exceeds a certain threshold $a_n$. Therefore, $\mathcal{B}_n$ can be viewed as set that indicates the number of nonzero predictor coefficients.
Further suppose $\pi_{n}(\cdot)$ and $\Pi_n(\cdot)$ are the prior and posterior densities of $\bbeta_{n}$ with $n$ observations respectively, so that
\begin{align*}
\pi_n(\bbeta_n)=\prod_{j=1}^{p_n}h_{\mu_n}(\beta_{j,n}),\:\:
\Pi_n(\mathcal{C}_n)=\frac{\int_{\mathcal{C}_n}f(\tilde{\by}_n|\bbeta_n)\pi_{n}(\bbeta_n)}
                          {\int f(\tilde{\by}_n|\bbeta_n)\pi_{n}(\bbeta_n)},
\end{align*}
where $f(\tilde{\by}_n|\bbeta_{n})$ is the joint density of $\tilde{\by}_n=\bPhi_n \by_n$ under model (\ref{model1}).
%This article intends to show
%\begin{align}\label{consistency1}
%\Pi_n(\mathcal{A}_n)\rightarrow 0,\:\:\mbox{a.s., when}\:\:n\rightarrow\infty.
%\end{align}
The following theorem shows posterior contraction for the proposed model, with the proof of the theorem given in the appendix.
\begin{theorem}\label{thm_main}
Under Assumptions (A)-(G), our proposed model satisfies $E_{\bbeta_n^*}(\Pi_n(\mathcal{C}_n))\rightarrow 0$, as $n,m_n\rightarrow\infty$ with the posterior contraction rate $\theta_n=E\sqrt{s_n\log(p_n)/m_n}$, for some constant $ E>0$.
\end{theorem}
The general result on posterior contraction in Theorem~\ref{thm_main} is applied to provide posterior contraction result for the proposed data sketching approach with a class of Gaussian scale mixture prior distributions on $\beta_{j,n}$. Indeed we assume that the prior density $h_{\mu_n}$ with hyper-parameter $\mu_n$ of each $\beta_{j,n}$ is symmetric around $0$ and has a polynomial tail, i.e., $h_{\mu_n}(x)\sim x^{-r}$ when $|x|$ is large, for some $r>1$. Notably prior densities for both the horseshoe shrinkage prior \citep{carvalho2010horseshoe} and the generalized double pareto shrinkage prior \citep{armagan2013generalized} have polynomial tails. Theorem~\ref{thm_main} can be adapted in
such a setting to arrive at the following result. The proof of the result can be found in the Appendix.
\begin{theorem}\label{consistency_new}
Let the feature matrix $\bX_n$, random compression matrix $\bPhi_n$ and the true feature coefficients $\bbeta_n^*$ satisfy Assumptions (A)-(G). Let the prior density with hyper-parameter $\mu_n$, given by $h_{\mu_n}(x)=h(x/\mu_n)$, has a polynomial tail, i.e., $h_{\mu_n}(x)\sim x^{-r}$ when $|x|$ is large, for some $r>1$. Further assume that  $\log(M_n)=O(\log p_n)$, $a_n=\sqrt{s_n\log(p_n)/m_n}/p_n$, $\mu_n\leq a_np_n^{-(u'+1)/(r-1)}$ and $\log(\mu_n)=O(\log(p_n))$, for some $u'>0$. Then the posterior contraction rate $\theta_n$ can be taken as
$E\sqrt{s_n\log(p_n)/m_n}$,for some constant $E>0$.
\end{theorem}
Note that the minimax optimal posterior contraction rate without data sketching is given by $\sqrt{s_n\log(p_n)/n}$ which is $\rho_{n}=\sqrt{n/m_n}$ times faster that the posterior contraction rate with data sketching. In fact, $\rho_n$ throws light on the connection between the theoretical performance of (\ref{model1}) with the choice of $m_n$.
In particular, choice of $m_n=O(n/\log(n)^{\tilde{K}})$ maintains minimax optimal posterior contraction rate upto a $\log(n)$ factor even with data sketching. The next section empirically studies the performance of data sketching in high dimensional regressions with various other competitors. Special emphasis is given to investigate the discrepancy in the inference on $\bbeta_n$ from the full data and the sketched data to carefully assess the impact of sketching.
%%discuss the choice of $m_n$.
\section{Simulation Studies}\label{sec4}
This section investigates performance of our data sketching approach (\ref{model1}) with the horseshoe shrinkage prior \citep{carvalho2010horseshoe} on each of the predictor coefficients $\beta_j$, referred to as the Compressed Horseshoe (CHS). Broadly, we implement and present two different sets of simulations. In \textbf{Simulation 1}, we focus on data simulated from (\ref{eq:GLM}) with $n=1000$ and $p=10000$, where both models (\ref{eq:GLM}) and (\ref{model1}) can be fitted to analyze the difference in their posterior distributions of $\bbeta$ for different choices of $m$ and different degrees of sparsity. These simulation examples also highlight the relative computational efficiency of (\ref{model1}) with respect to (\ref{eq:GLM}). \textbf{Simulation 2} is then designed with a larger sample size $n=5000$ and $p=10000$ which render infeasibility in fitting the model (\ref{eq:GLM}) with the uncompressed data based on our available computational resources. Thus the purpose for \textbf{Simulation 2} is to assess the frequentist operating characteristics of CHS along with a few of its frequentist competitors.

\subsection{Simulation 1: comparison between the performances of CHS and HS for moderate $n$ and large $p$}
\noindent In \textbf{Simulation 1}, we draw $n=1000$ samples from the high dimensional linear regression model (\ref{eq:GLM}) with the number of features $p=10000$ and the error
variance $\sigma^2=1.5$. The $p$-dimensional feature vectors $\bx_i$ for each $i=1,...,n$ are simulated from $N(\bzero,\bSigma)$, with two different constructions of $\bSigma$  undertaken in simulation studies.\\
\emph{Scenario 1:} $\bSigma=\bI_p$, i.e., all features are simulated i.i.d. We refer to this as the independent correlation structure for the features.\\
\emph{Scenario 2:} $\bSigma=0.5 \bI_p+0.5 \bJ_p$, where $\bJ_p$ is a matrix with $1$ at each entry. This structure ensures that any pair of features have the same correlation of $0.5$. We refer to this as the compound correlation structure for the features.\\
Under Scenarios 1 and 2, the $p$-dimensional true feature coefficient vector is simulated with the number of nonzero entries: (a) $s=10$; (b) $s=30$ and (c) $s=50$. The quantity $(1-s/p)$ is referred to as the true sparsity of the model. The magnitude of $s$ nonzero entries are simulated randomly from a $U(1.5,3)$ distribution with the sign of each entry randomly assigned to be positive or negative.

To compare the effect of data sketching on the estimation of posterior distribution of $\bbeta$, we implement (\ref{eq:GLM}) (with the uncompressed data) and (\ref{model1}) with $m=100,200,300,400,500$. The full/uncompressed data posterior distribution obtained using MCMC serves as the benchmark in
our assessment of the performance of (\ref{model1}). Let $\pi(\beta_j|\by,\bX)$ be the density of the full data posterior distribution for $\beta_j$ estimated using sampling and $\pi_m(\beta_j|\tilde{\by},\tilde{\bX})$ be the density of posterior distribution for $\bbeta$ estimated using (\ref{model1}) with the compressed data, where the subscript $m$ denotes the dimension of the sketching matrix $\bPhi$ to compute $\tilde{\by}$ and $\tilde{\bX}$. We used the following metric
based on the Hellinger distance to compare the accuracy of $\pi_m(\beta_j|\tilde{\by},\tilde{\bX})$ in approximating $\pi(\beta_j|\by,\bX)$
\begin{align}\label{KL_div}
Accuracy_{j,m}=1-\frac{1}{2}\int_{\bbeta}\left(\sqrt{\pi_m(\beta_j|\tilde{\by},\tilde{\bX})}-\sqrt{\pi(\beta_j|\by,\bX)}\right)^2d\beta_j.
\end{align}
The metric Accuracy$_{j,m}$ satisfies $0\leq$ Accuracy$_{j,m}$ $\leq 1$.
The approximation of full data posterior
density $\pi(\beta_j|\by,\bX)$ by $\pi_m(\beta_j|\tilde{\by},\tilde{\bX})$ is poor or excellent if the accuracy metric is close to 0 or 1, respectively. We present
Accuracy$_{j,m}$ averaged over all predictors, given by Accuracy$_m=\frac{1}{p}\sum_{j=1}^p$ Accuracy$_{j,m}$.

\noindent \textbf{Simulation 1} also highlights the computational efficiency offered by the data sketching approach. Let ESS$_{m}$ be the average effective sample size of $\bbeta$ (out of $5000$ post burn-in iterates) from (\ref{model1}) with rank($\bPhi)=m$, that runs for $T_{m}$ hours. We will measure the computational efficiency of our proposed approach for a specific choice of $m$ as
\begin{align}
  \label{eq:s3}
  \text{Computational Efficiency}_{m} = log_2 \text{ESS}_{m} / \text{T}_{m},
\end{align}
where ESS$_{m}$ is the effective sample size over $p$ feature coefficients computed using the coda R package. Computational efficiency of the full posterior will also be reported
to provide a relative assessment. All simulations are replicated $50$ times.

\subsubsection{Results}\label{sec:small}
Table~\ref{accuracy} presents the Accuracy metric averaged over all predictors and all replications. The results show excellent performance of $\pi_m(\bbeta|\tilde{\by},\tilde{\bX})$ in approximating $\pi(\bbeta|\by,\bX)$ for all cases except when both the sparsity and rank of the random compression matrix
is low. This empirical observation is also supported by Theorem ~\ref{thm_main} which requires the degree of sparsity to grow at a much slower rate than the rank
of the random compression matrix. Understandably, as $m$ increases the accuracy becomes close to $1$, with the accuracy being little impacted when the sparsity is very low.
No notable difference is observed in the performance when predictors are correlated vis-a-vis when predictors are simulated independently.
\begin{table}[h]
\centering
\begin{tabular}{|  l | l | | c | c | c | c | c | c |}
\cline{1-8}
 & &\multicolumn{3}{ |c| }{Scenario 1} &  \multicolumn{3}{ |c| }{Scenario 2}\\
\cline{2-8}
& & $s=10$ & $s=30$ & $s=50$ & $s=10$ & $s=30$ & $s=50$ \\
\hline
\multirow{5}{*}{Avg. Accuracy}
 %$m=100$ & $s=10$ & $s=30$ & $s=50$ & $s=10$ & $s=30$ & $s=50$ \\ \hline
 & $m=100$ & 0.88 & 0.79 & 0.64 & 0.88 & 0.81 & 0.66\\
 & $m=200$ & 0.94 & 0.87 &  0.76 & 0.92 & 0.84 & 0.73  \\
 & $m=300$ & 0.98 & 0.96 & 0.93 & 0.99 & 0.96 & 0.94 \\
 & $m=400$ & 0.98 & 0.98 & 0.94 & 0.98 & 0.98 & 0.94\\
 & $m=500$ & 0.98 & 0.98 & 0.95 & 0.99 & 0.98 & 0.95\\
\hline
 \multirow{5}{*}{Comp. Efficiency} & $m=100$ & 2.83 & 2.81 & 2.81 & 2.86 & 2.81 & 2.83\\
 & $m=200$ & 2.01 & 2.03 & 2.03 & 2.02 & 2.04 & 2.03\\
 & $m=300$ & 1.28 & 1.30 & 1.30 & 1.32 & 1.31 &  1.32\\
 & $m=400$ & 1.03 & 1.02 & 1.06 & 1.03 & 1.02 & 1.05\\
 & $m=500$ & 0.85 & 0.86 & 0.86 & 0.86 & 0.87 & 0.84\\
 & HS      & 0.32 & 0.31 & 0.31 & 0.31 & 0.32 & 0.32\\
 \hline
\end{tabular}
\caption{The first five rows present metric to estimate accuracy of estimating full posterior of $\bbeta$ by the posterior of $\bbeta$ with compressed data, as described in (\ref{KL_div}). We present the metric averaged over all predictors and all replications. The metric is presented for different choices of $m=100,200,300,400,500$ and different degrees of sparsity for the true coefficient $\bbeta^*$. The upper bound of the accuracy measure is 1 and a higher value represents more accuracy. We also present computational efficiency of CHS, as described in (\ref{eq:s3}), for different choices of $m$ and for different degrees of sparsity under the two different simulation scenarios. Computational efficiency of the posterior distribution of $\bbeta$ with the uncompressed data (referred to as the HS) has also been presented.}\label{accuracy}
\end{table}
Since the sample size is moderate, we do not expect to see a lot of gain in terms of computational efficiency of CHS over HS. CHS with $m=100$ appears to be around $\sim 10$ times
computationally more efficient than HS. The computational efficiency decreases as we increase the rank of the compression matrix. Similar to accuracy, the computational efficiency seems to be not severely affected by the degree of sparsity or the correlation in the features.
\begin{comment}
\begin{table}[h]
\centering
\begin{tabular}{|  l | l | c | c | c | c | c | c |}
\cline{1-7}
 & & \multicolumn{3}{ |c| }{Scenario 1} &  \multicolumn{3}{ |c| }{Scenario 2}\\
\cline{2-7}
 \hline
 \multirow{5}{*}{CHS} & $m=100$ & 1.43 &  & & 1.46 &  & \\
 & $m=200$ &  & &  &  & & \\
 & $m=300$ &  & &  &  & &  \\
 & $m=400$ &  & &  &  & & \\
 & $m=500$ &  & &  &  & & \\
 \hline
 \multirow{5}{*}{HS} & $m=100$ & 0.09 &  &   & 0.09 & &  \\
 & $m=200$ &  & &  &  & & \\
 & $m=300$ &  & &  &  & &  \\
 & $m=400$ &  & &  &  & & \\
 & $m=500$ &  & &  &  & & \\
 \hline
\end{tabular}
\caption{Mean Squared Error (MSE), average coverage and average length of 95\% credible interval for SGTM, LS, HOLRR and ENV are presented for cases under simulation Scenario 1 for the continuous case. The lowest MSE in each case is boldfaced. Results are averaged over $50$ replications and the standard deviations over $50$ replications are reported in the subscript of every number.}\label{tab2_scen1}
\end{table}
\end{comment}

\subsection{Simulation 2: comparison between CHS and its frequentist competitors with larger sample size}
\noindent \textbf{Simulation 2} is designed to assess performance of the proposed framework for a large $p$, large $n$ setting. Thus, we follow the identical data generation scheme as \textbf{Simulation 1} with a large sample size $n=5000$ to construct simulated data. The large values of $p$ and $n$ prohibits Bayesian model fitting of (\ref{eq:GLM})
using the horseshoe prior using our available computational resources. Hence, we focus on investigating frequentist operating characteristics of CHS along with its frequentist competitors in high dimensional regression. As a frequentist competitor to CHS, we implement the minimax concave penalty (MCP) method \citep{zhang2010nearly} on the full data.
Additionally, we fit MCP on randomly chosen $m$ data points from the sample of size $n$, and refer to this competitor as Partial MCP (PMCP). The MCP on full data provides a comparison of our approach with a frequentist penalized optimizer in high dimensional regression with big $n$ and $p$. While MCP with the full data is likely to perform better than CHS with the compressed data, the discrepancy in performance of CHS and MCP can be seen as an indicator of loss of inference due to data compression. On the other hand, comparison of CHS with PMCP demonstrates the inferential advantage of fitting a principled Bayesian approach with sketching that uses information from the entire sample over fitting of a frequentist penalization scheme with naive sampling of $m$ out of $n$ data points. Although the remaining section presents excellent performance of the sketching approach with the horseshoe prior on $\beta_j$'s, we expect similar performance from other Gaussian scale mixture prior distributions, such as the Generalized Double Pareto \citep{armagan2013generalized} prior or the normal gamma prior \citep{griffin2010inference}.

%In the simulation examples, we draw $n=5000$ samples from the high dimensional linear regression model (\ref{eq:GLM}) with the number of predictors $p=10000$ and the error
%variance $\sigma^2=1.5$. The $p$-dimensional predictor vectors $x_i$ for each $i=1,...,n$ are simulated from $N(0,\Sigma)$, with two different constructions of $\Sigma$  undertaken in simulation studies.\\
%\emph{Scenario 1:} $\Sigma=I_p$, i.e., all predictors are simulated i.i.d. We refer to this as the independent correlation structure for the predictors.\\
%\emph{Scenario 2:} $\Sigma=0.5 I_p+0.5 J_p$, where $J_p$ is a matrix with $1$ at each entry. This structure ensures that any pair of predictors have the same correlation of $0.5$. We refer to this as the compound correlation structure for the predictors.\\
%Under Scenarios 1 and 2, the $p$-dimensional true predictor coefficient vector is simulated with the number of nonzero entries: (a) $s=10$; (b) $s=30$ and (c) $s=50$. The quantity $(1-s/p)$ is referred to as the true sparsity of the model. The magnitude of $s$ nonzero entries are simulated randomly from a $U(1.5,3)$ distribution with the sign of each entry randomly assigned to be positive or negative.

To assess how the true sparsity $(1-s/p)$ and the rank $m$ of the random compression matrix interplay, we fit CHS with $m=200$ and $m=400$ in both simulation scenarios under the three different sparsity levels corresponding to (a), (b) and (c). For MCMC based model implementation of CHS, we discard the first $5000$ samples as burn-in and draw inference based on the $5000$ post burn-in samples. Both MCP and PMCP are fitted with the \texttt{R} package \texttt{ncvreg} with tuning parameters chosen using a $10$-fold cross validation.

The inferential performances of the competitors are compared based on
the overall mean squared error (MSE) of estimating the true predictor coefficient vector $\bbeta^*$ and the mean squared error of estimating the truly nonzero predictor coefficient vector $\bbeta_{nz}^*$ (referred to as the MSE$_{nz}$). These metrics are given by
\begin{align}\label{metric}
\mbox{MSE}=||\hat{\bbeta}-\bbeta^*||_2^2/p,\:\:\:\:\mbox{MSE}_{nz}=||\hat{\bbeta}_{nz}-\bbeta_{nz}^*||_2^2/s,
\end{align}
where $\hat{\bbeta}$ and $\hat{\bbeta}_{nz}$ is a point estimate for $\bbeta$ and $\bbeta_{nz}$, respectively. For CHS, the point estimate is taken to be the posterior mean.
Uncertainty of estimating $\bbeta$ from CHS is characterized through coverage and length of 95\% credible intervals averaged over all $\beta_j$'s, $j=1,...,p$. Additionally, we report the coverage and length of 95\% credible intervals averaged over truly nonzero $\beta_j$'s. Since model fitting in (\ref{model1}) is performed with data sketches, it is not possible to draw predictive inference directly. Hence, the quantity $||\bX\hat{\bbeta}-\bX\bbeta^*||_2^2/n$ is reported to provide a rough assessment of the predictive inference from CHS. This quantity is also computed and presented for other competitors.
%Notably, in all three cases (a)-(c), $s\log(p)$ is similar or larger than $m$, presenting challenging contexts where the theoretical guarantee (as given in Theorem~\ref{thm_main}) on estimation of $\beta$ do not necessarily follow.
All results presented are averaged over $50$ replications.

\subsubsection{Results}
Figures~\ref{Fig_MSE_ind} and \ref{Fig_MSE_com} present the boxplots for MSE and MSE$_{nz}$ for all competitors under the three different sparsity levels in Scenarios 1 and 2, respectively. Understandably, MCP applied on the full data is the best performer in all simulation cases. With  small to moderate value of the ratio $s/m$, CHS significantly outperforms PMCP, both in terms of MSE and
MSE$_{nz}$. This becomes evident by comparing the performances of CHS and PMCP for $m=400$ under all three cases (a)-(c) and for the case $m=200, s=10$. In fact when $s/m$ is small, CHS is also found to offer competitive performance with MCP (refer to the results under $m=400$). This observation is consistent with our findings in Section~\ref{sec:small}, where small values of $s/m$ shows little discrepancy between the full posterior of $\bbeta$ and posterior of $\bbeta$ under data compression.  As sparsity decreases and $s/m$ becomes higher, the performance gap between CHS and PMCP narrows. This is evident from both Figures~\ref{Fig_MSE_ind} and \ref{Fig_MSE_com}, corresponding to the case with $s=30, 50$ and $m=200$.
Consistent with the point estimation of $\bbeta$, Table~\ref{table3} shows notable advantage of CHS over PMCP in terms of predictive inference, especially with smaller $s/m$. MCP on the full data is naturally found to be the superior performer among the three. We observe a similar trend in the performance, both under Scenario 1 and 2.

\begin{figure}
  \begin{center}
    \subfigure[MSE of $\bbeta$: CHS, $m=200$]{\includegraphics[width=4.0 cm]{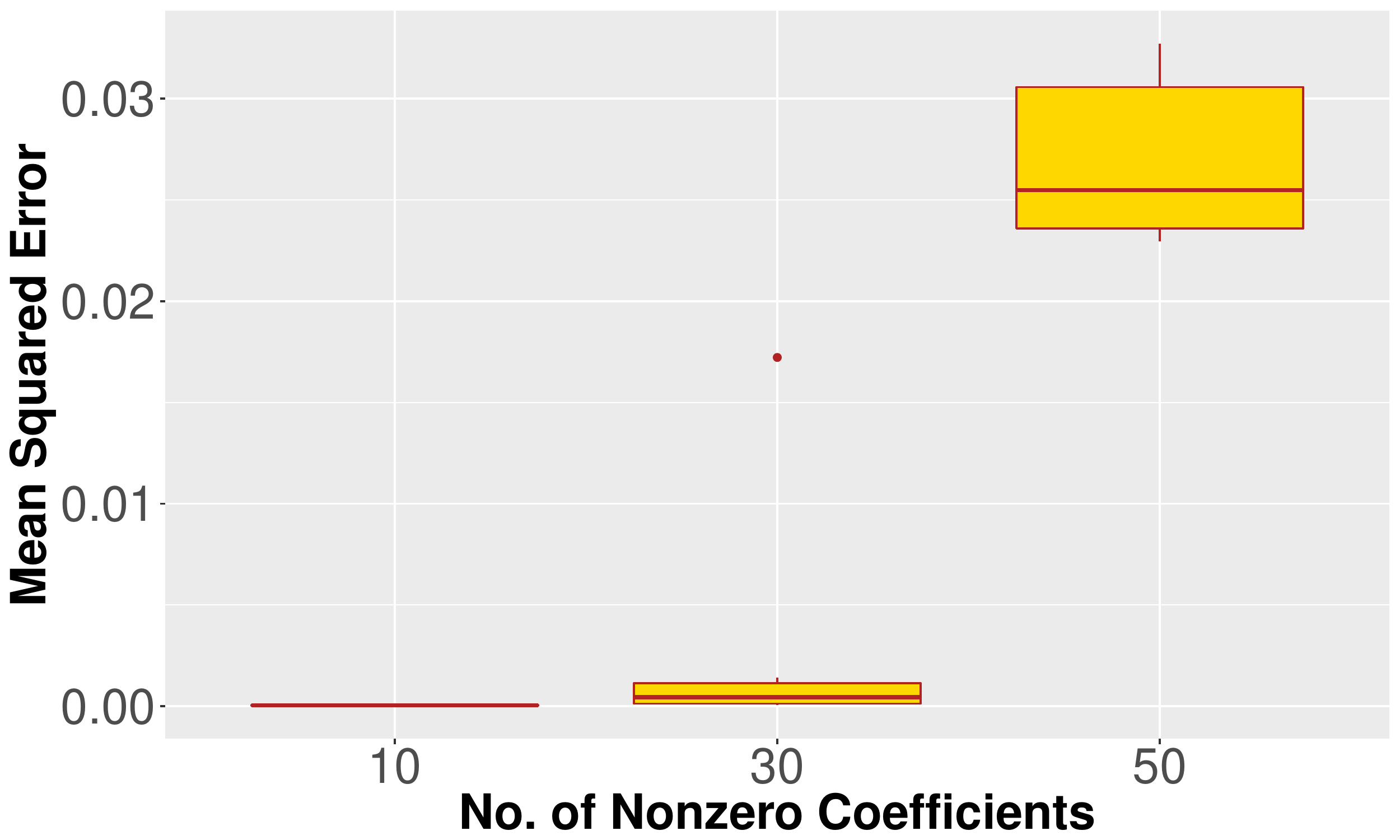}\label{AUC1}}
    \subfigure[MSE of $\bbeta$: PMCP, $m=200$]{\includegraphics[width=4.0 cm]{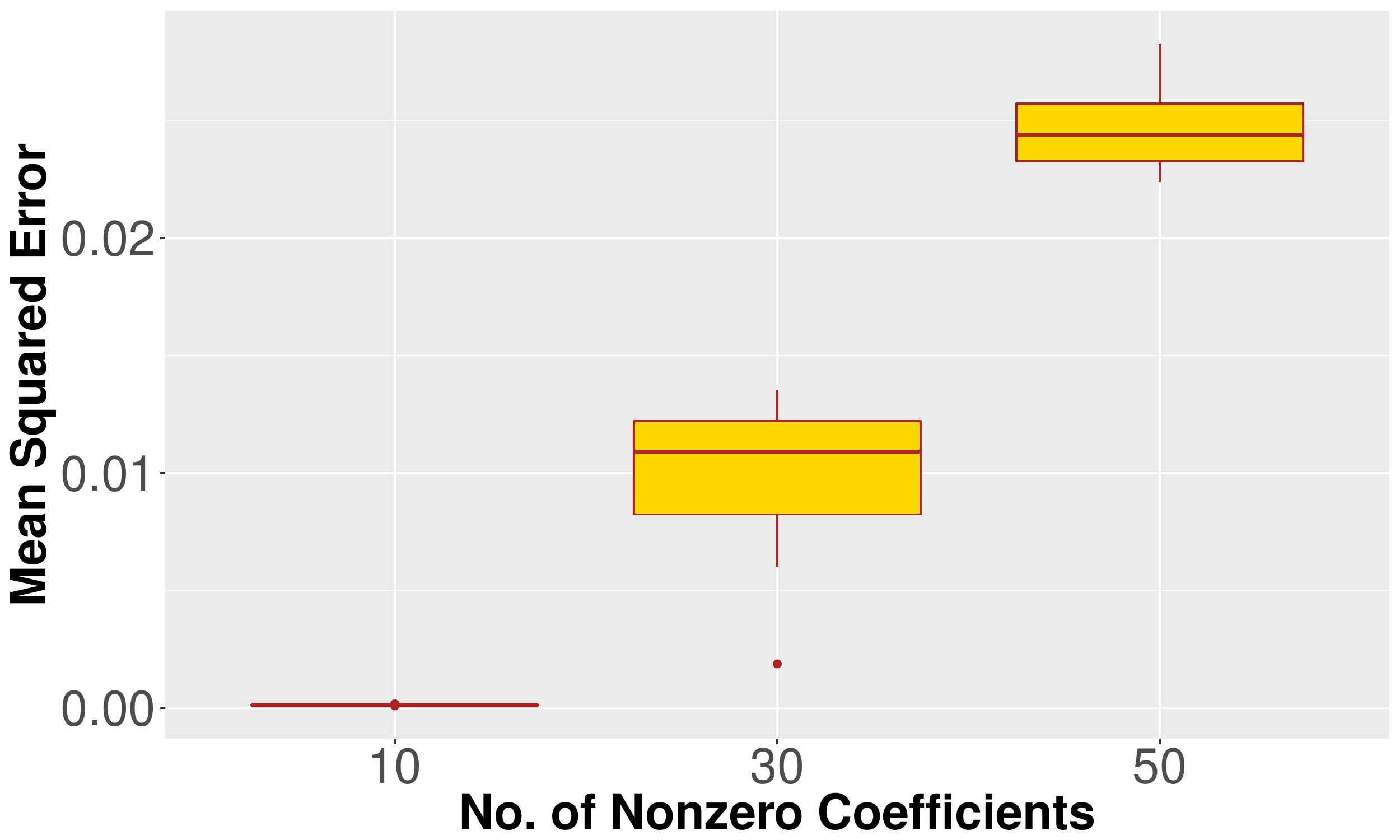}\label{AUC2}}
    \subfigure[MSE of $\bbeta$: MCP]{\includegraphics[width=4.0 cm]{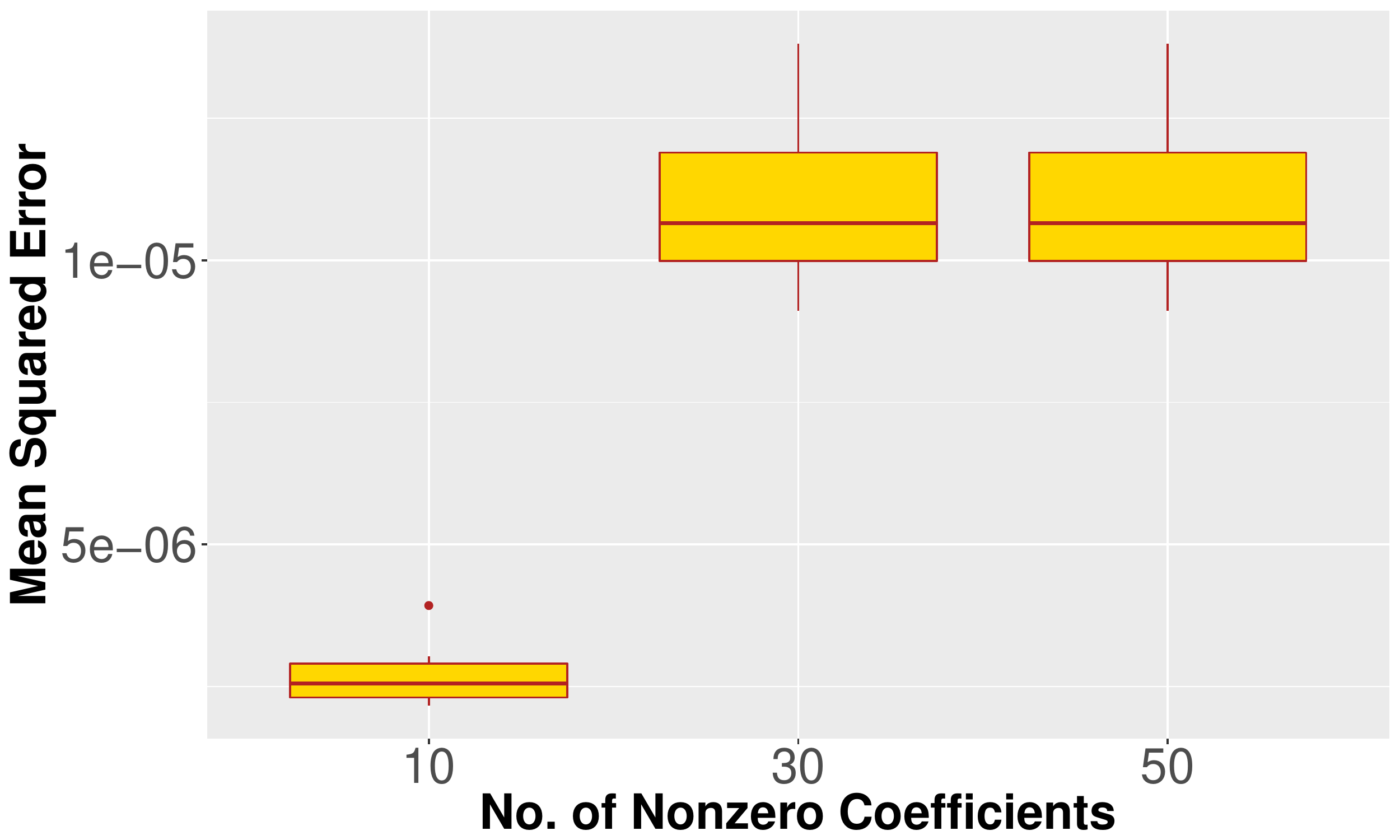}\label{AUC2}}\\
    \subfigure[MSE of nonzero $\bbeta$: CHS, $m=200$]{\includegraphics[width=4.0 cm]{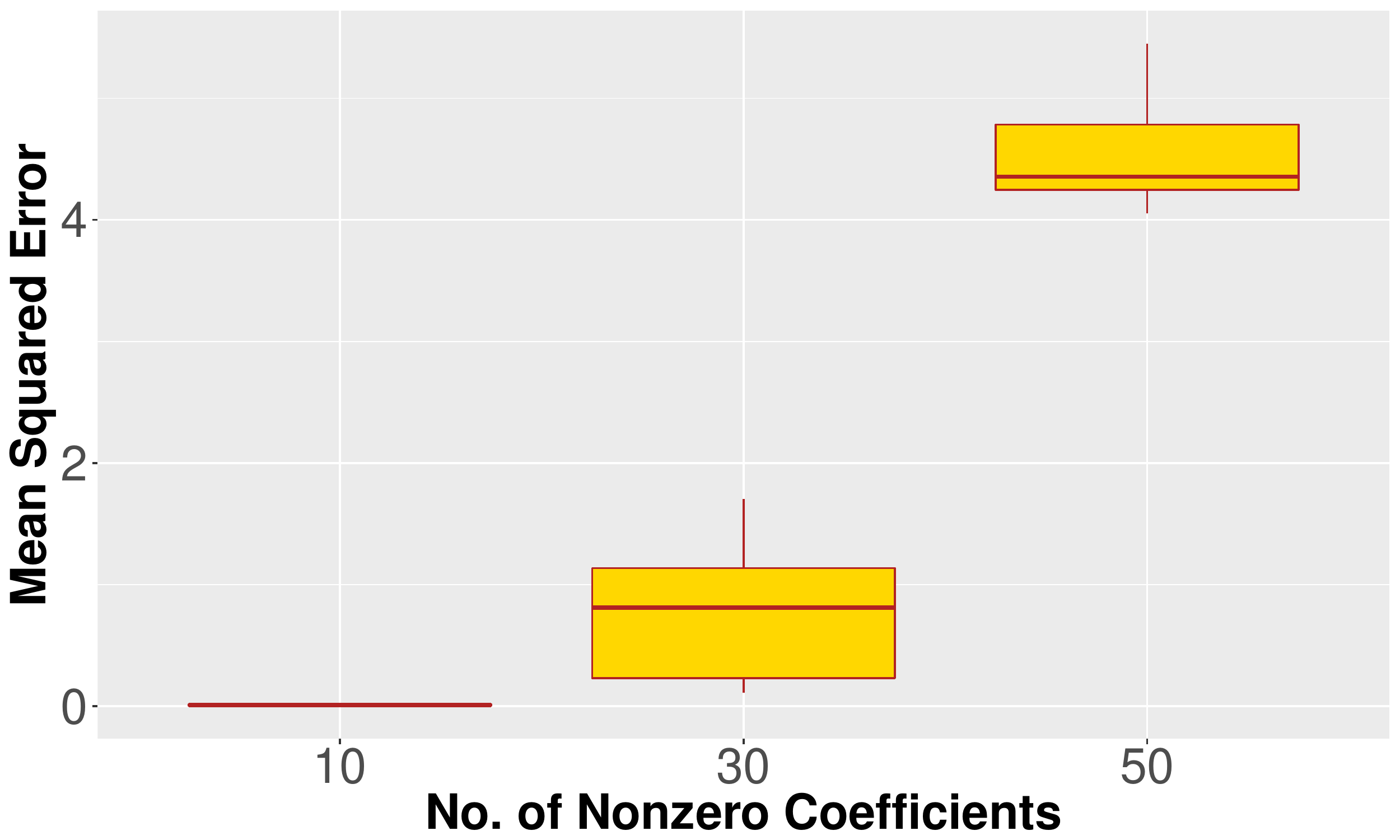}\label{AUC1}}
    \subfigure[MSE of nonzero $\bbeta$: PMCP, $m=200$]{\includegraphics[width=4.0 cm]{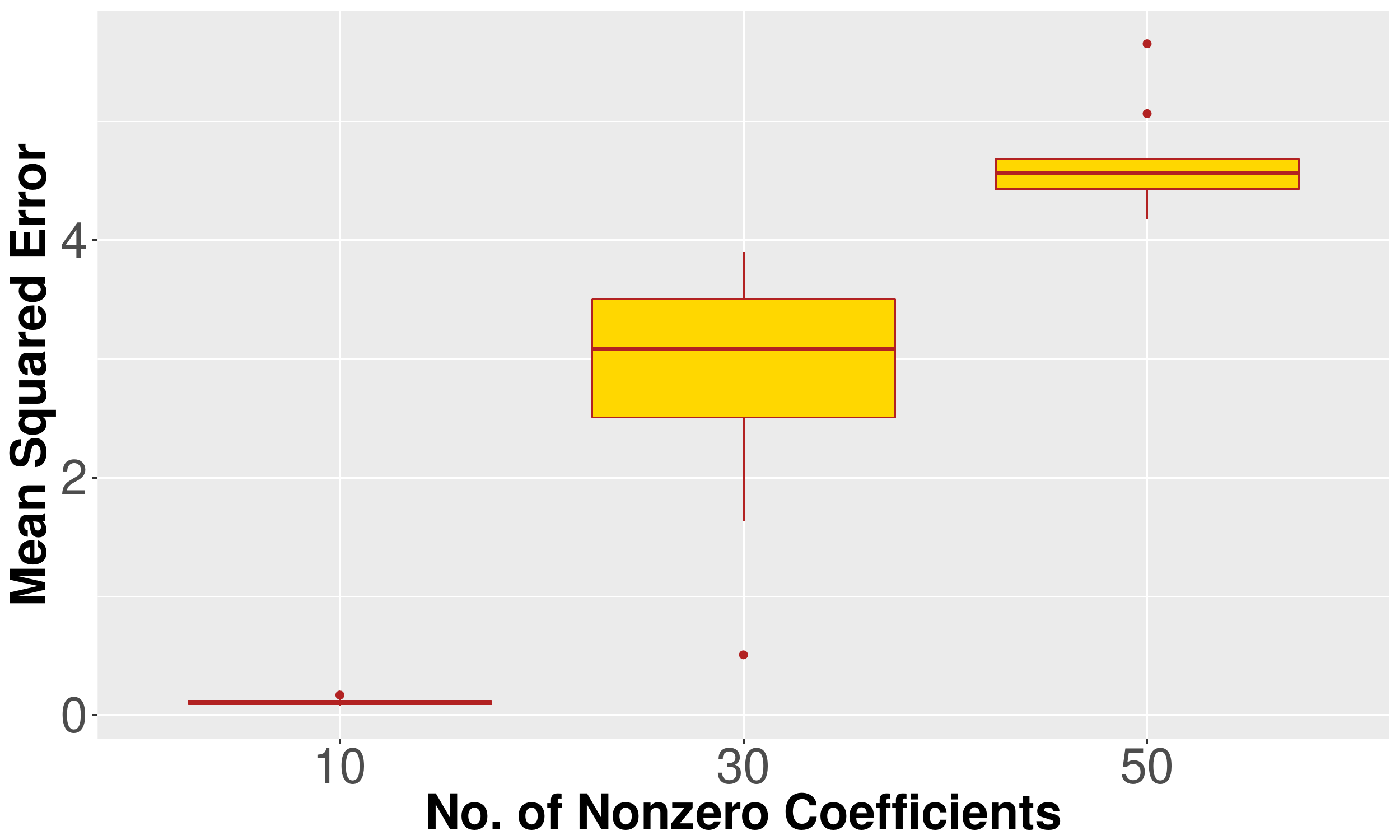}\label{AUC2}}
    \subfigure[MSE of nonzero $\bbeta$: MCP]{\includegraphics[width=4.0 cm]{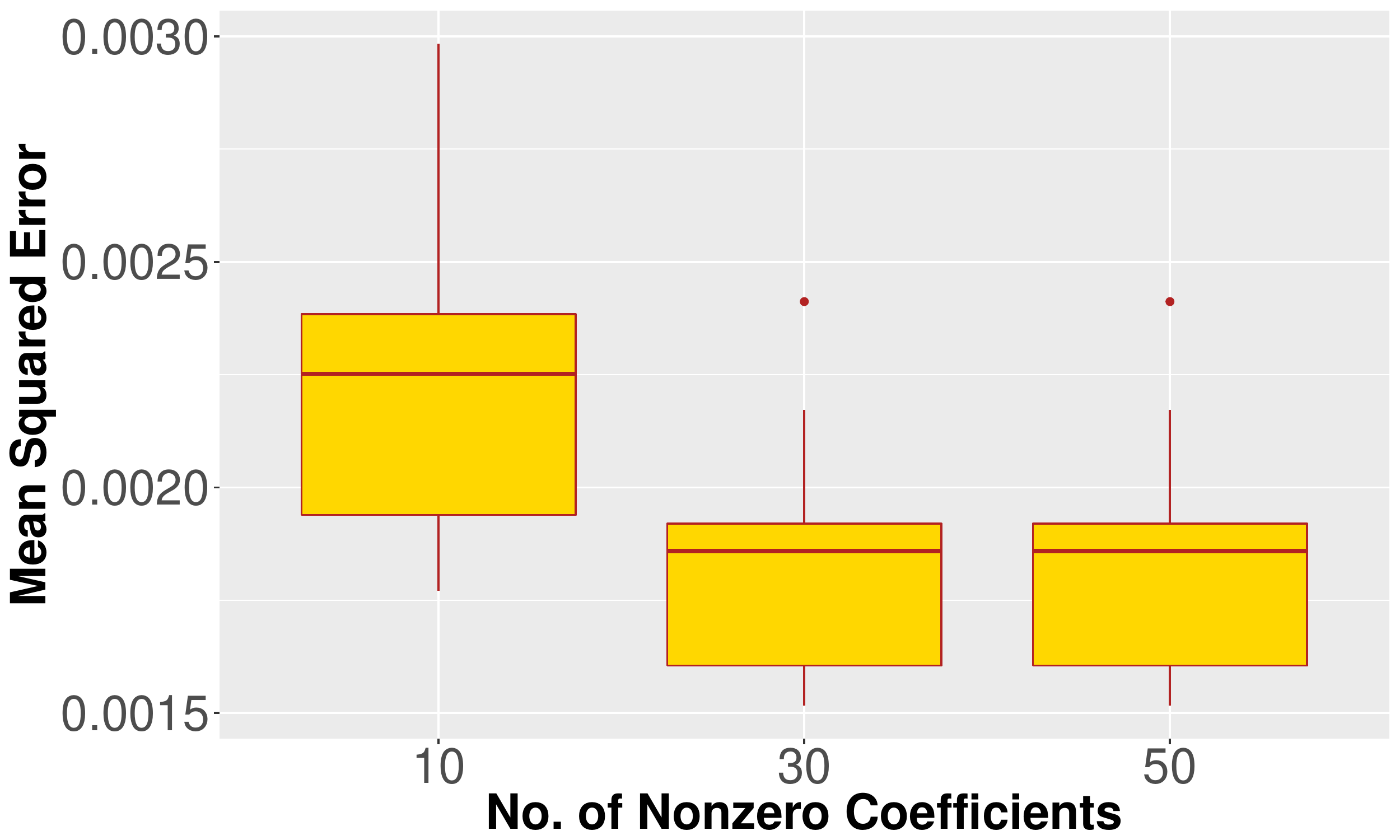}\label{AUC2}}\\
    \subfigure[MSE of $\bbeta$: CHS, $m=400$]{\includegraphics[width=4.0 cm]{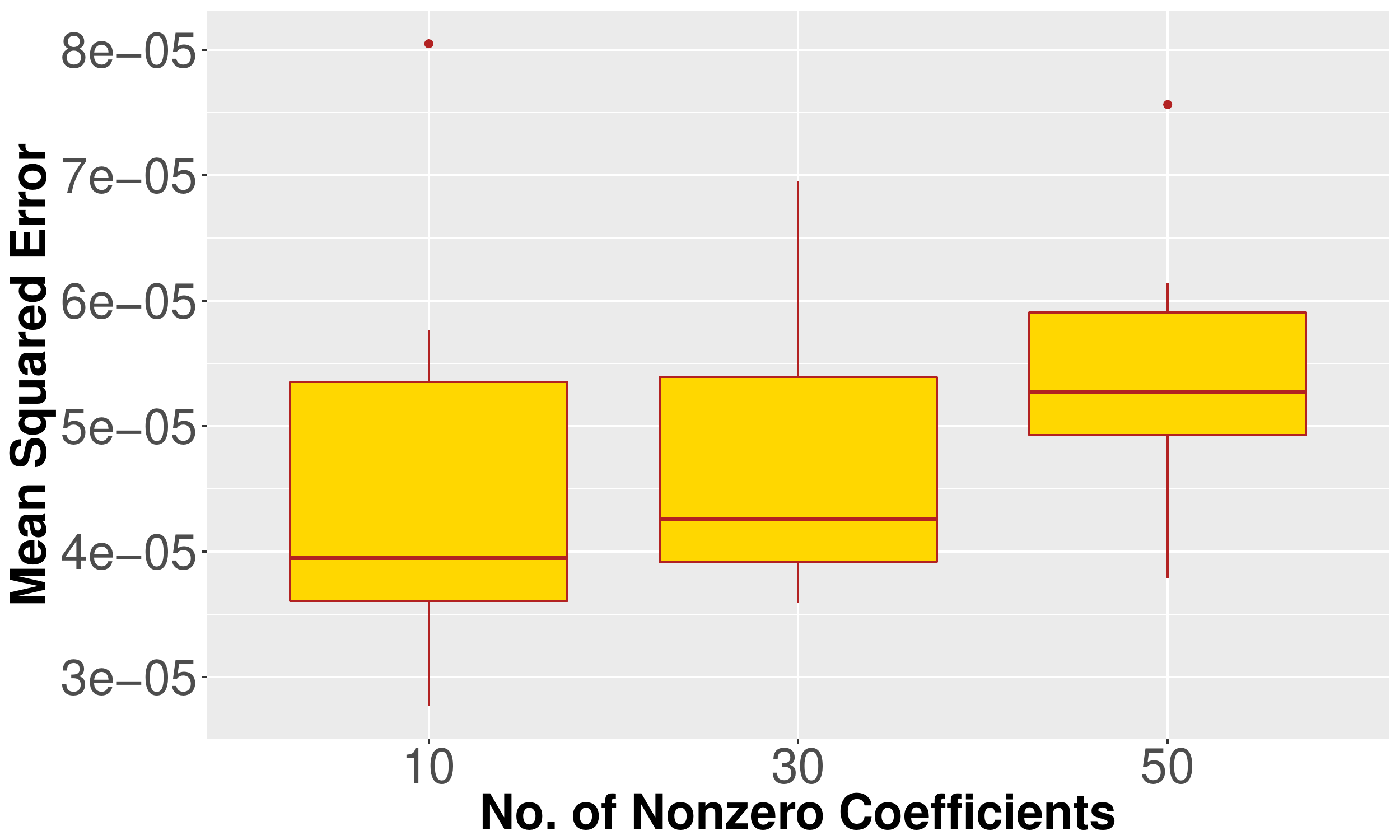}\label{AUC1}}
    \subfigure[MSE of $\bbeta$: PMCP, $m=400$]{\includegraphics[width=4.0 cm]{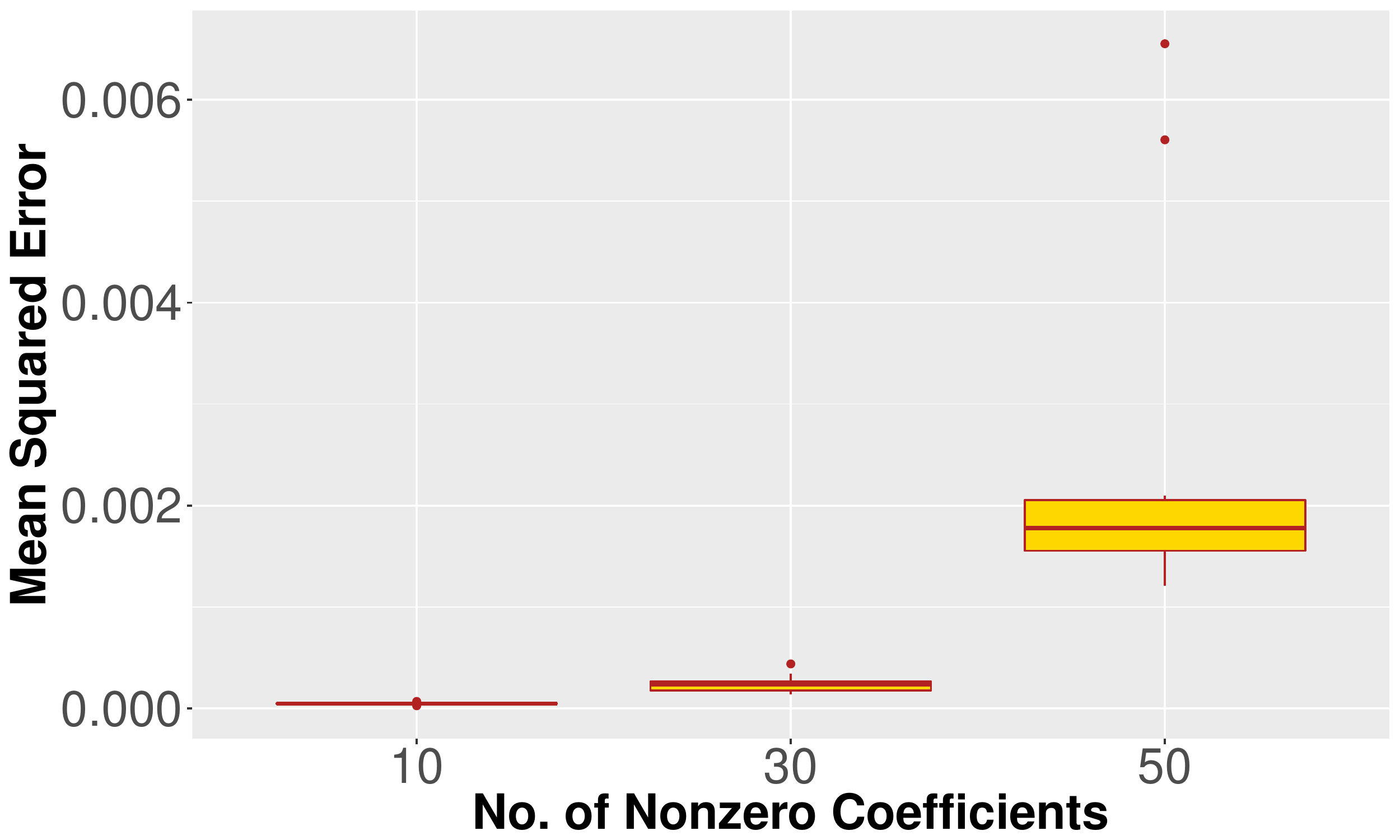}\label{AUC2}}
    \subfigure[MSE of $\bbeta$: MCP]{\includegraphics[width=4.0 cm]{MSE_m400_Lasso_ind.pdf}\label{AUC2}}\\
    \subfigure[MSE of nonzero $\bbeta$: CHS, $m=400$]{\includegraphics[width=4.0 cm]{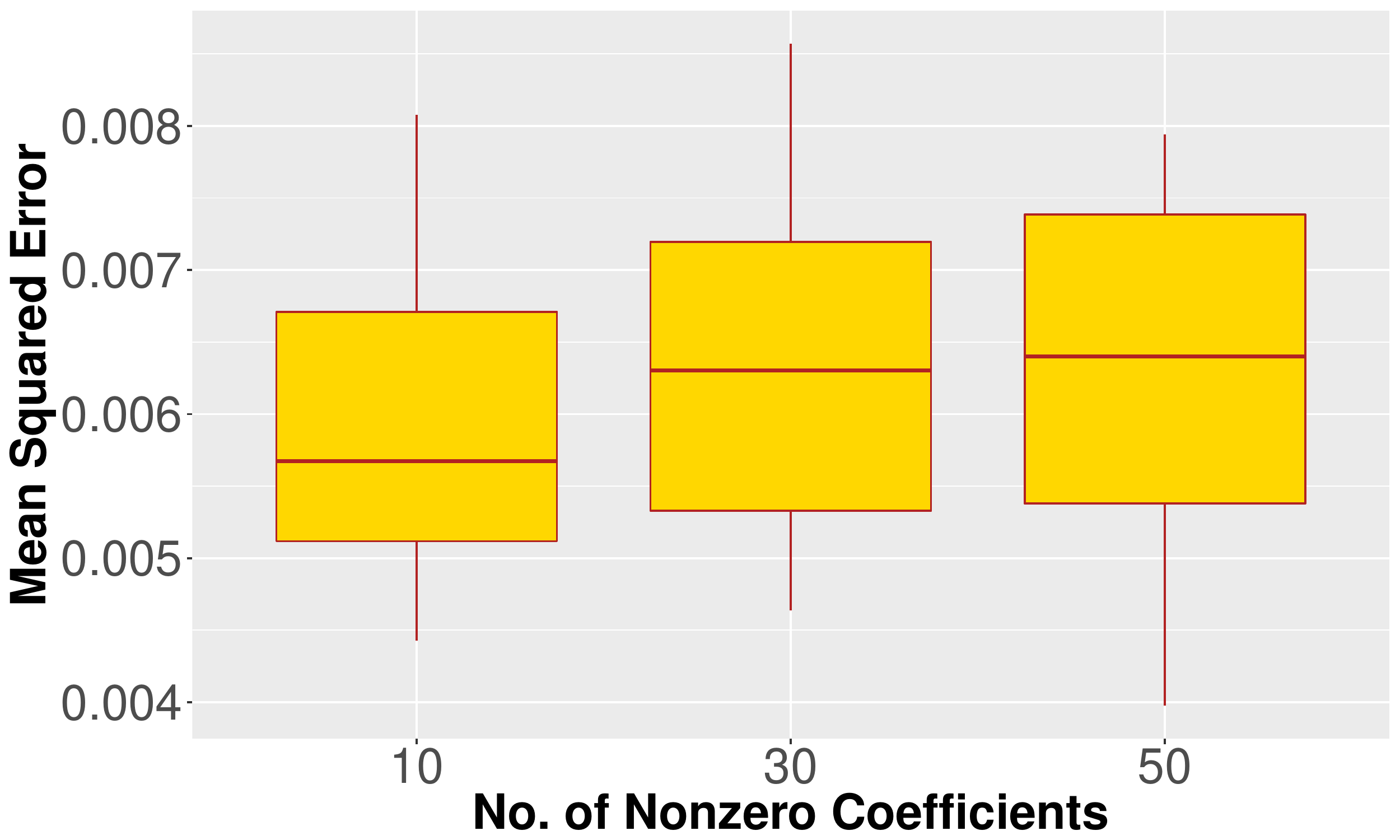}\label{AUC1}}
    \subfigure[MSE of nonzero $\bbeta$: PMCP, $m=400$]{\includegraphics[width=4.0 cm]{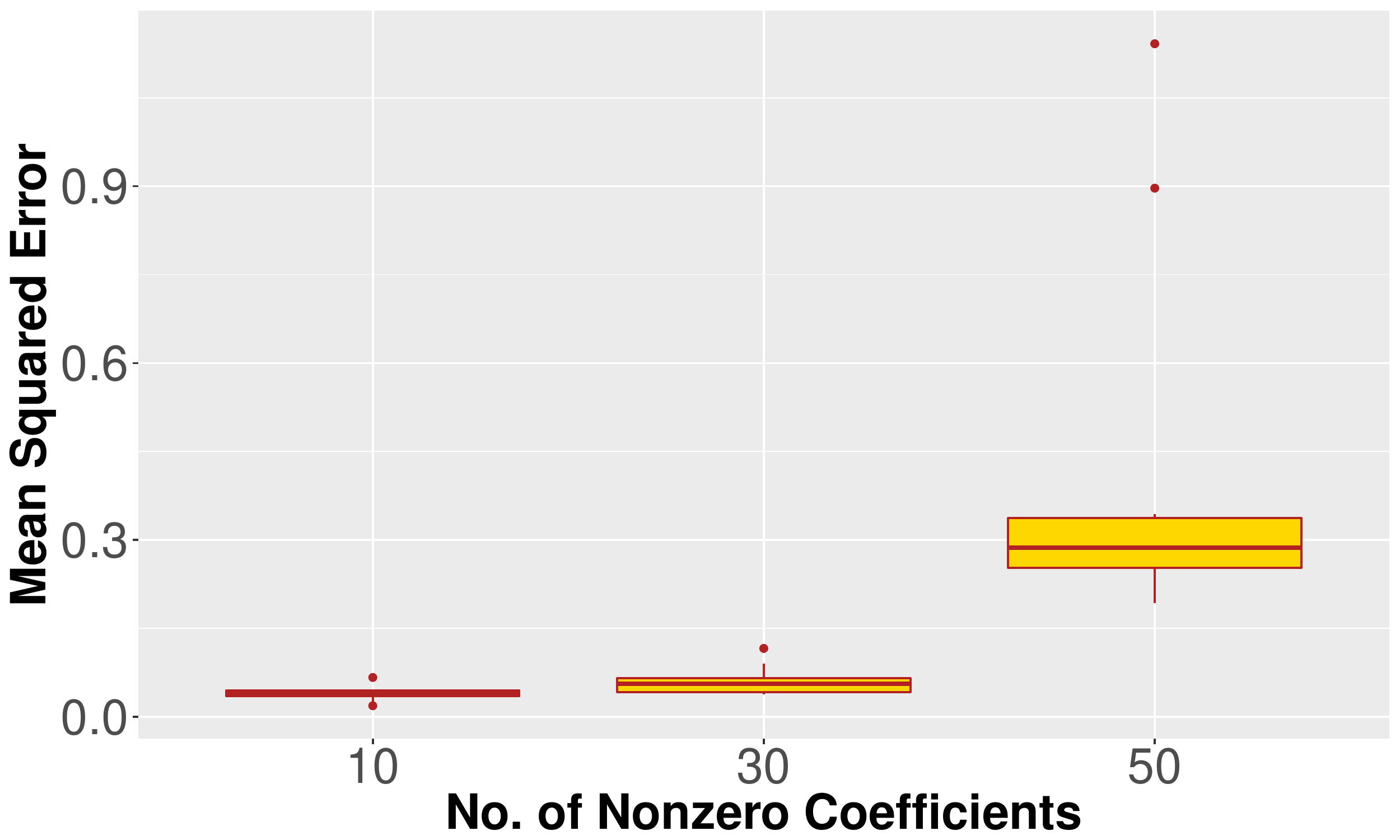}\label{AUC2}}
    \subfigure[MSE of nonzero $\bbeta$: MCP]{\includegraphics[width=4.0 cm]{MSE_nz_m400_Lasso_ind.pdf}\label{AUC2}}
% \end{center}
 \caption{First and third row present mean squared error (MSE) of estimating the true predictor coefficient $\bbeta^*$ by a point estimate of $\bbeta$ from CHS, PMCP and MCP for $m=200$ and $m=400$, respectively. Second and fourth row present mean squared error (MSE) of estimating the true nonzero coefficients in $\bbeta^*$ by a point estimate of the corresponding coefficients in $\bbeta$ from CHS, PMCP and MCP for $m=200$ and $m=400$, respectively. All figures correspond to the scenarios where the predictors are generated under the independent correlation structure (Scenario 1). Each figure shows performance of a competitor under the data generated with $10$, $30$ and $50$ nonzero coefficients in $\bbeta^*$. }\label{Fig_MSE_ind}
 \end{center}
\end{figure}

\begin{figure}
  \begin{center}
    \subfigure[MSE of $\bbeta$: CHS, $m=200$]{\includegraphics[width=4.0 cm]{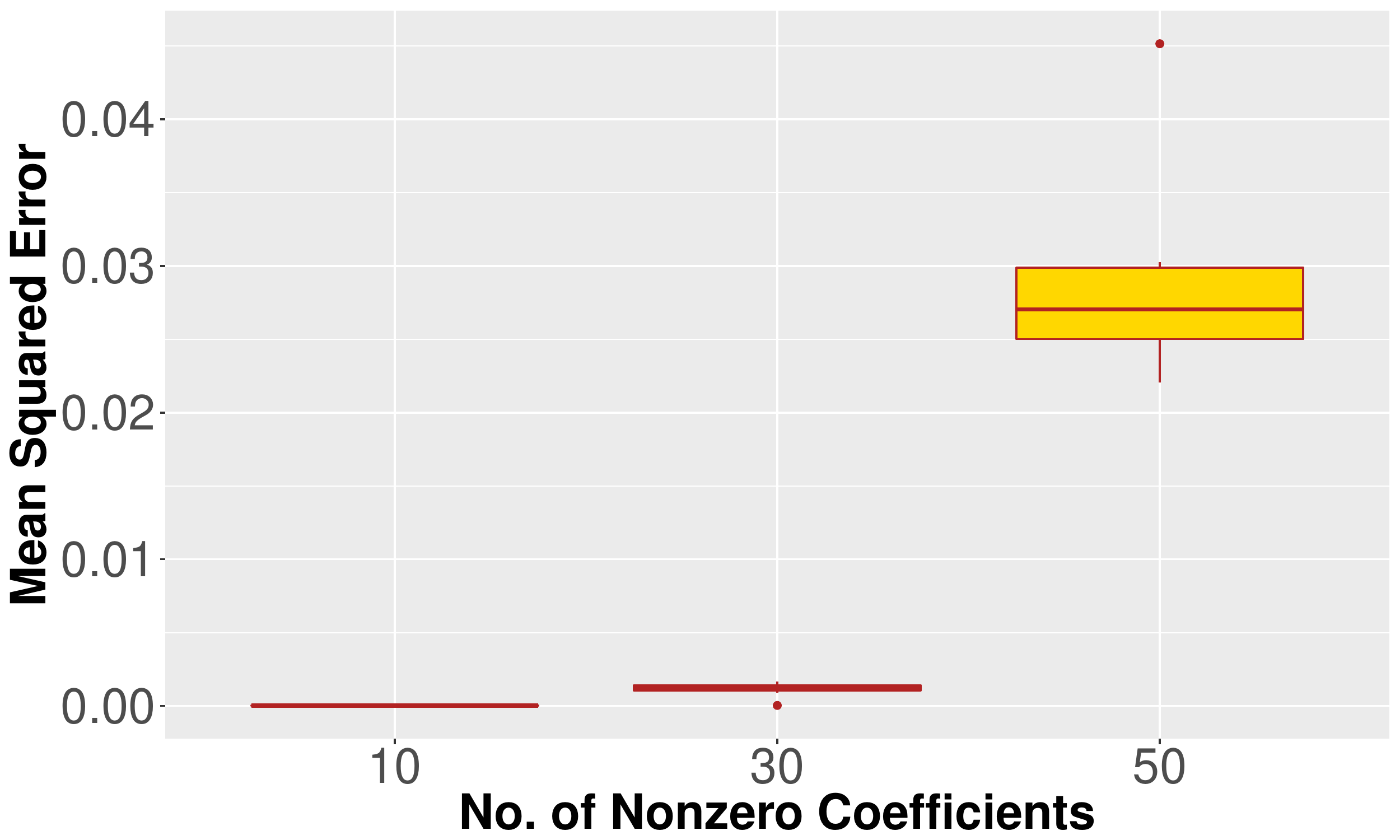}\label{AUC1}}
    \subfigure[MSE of $\bbeta$: PMCP, $m=200$]{\includegraphics[width=4.0 cm]{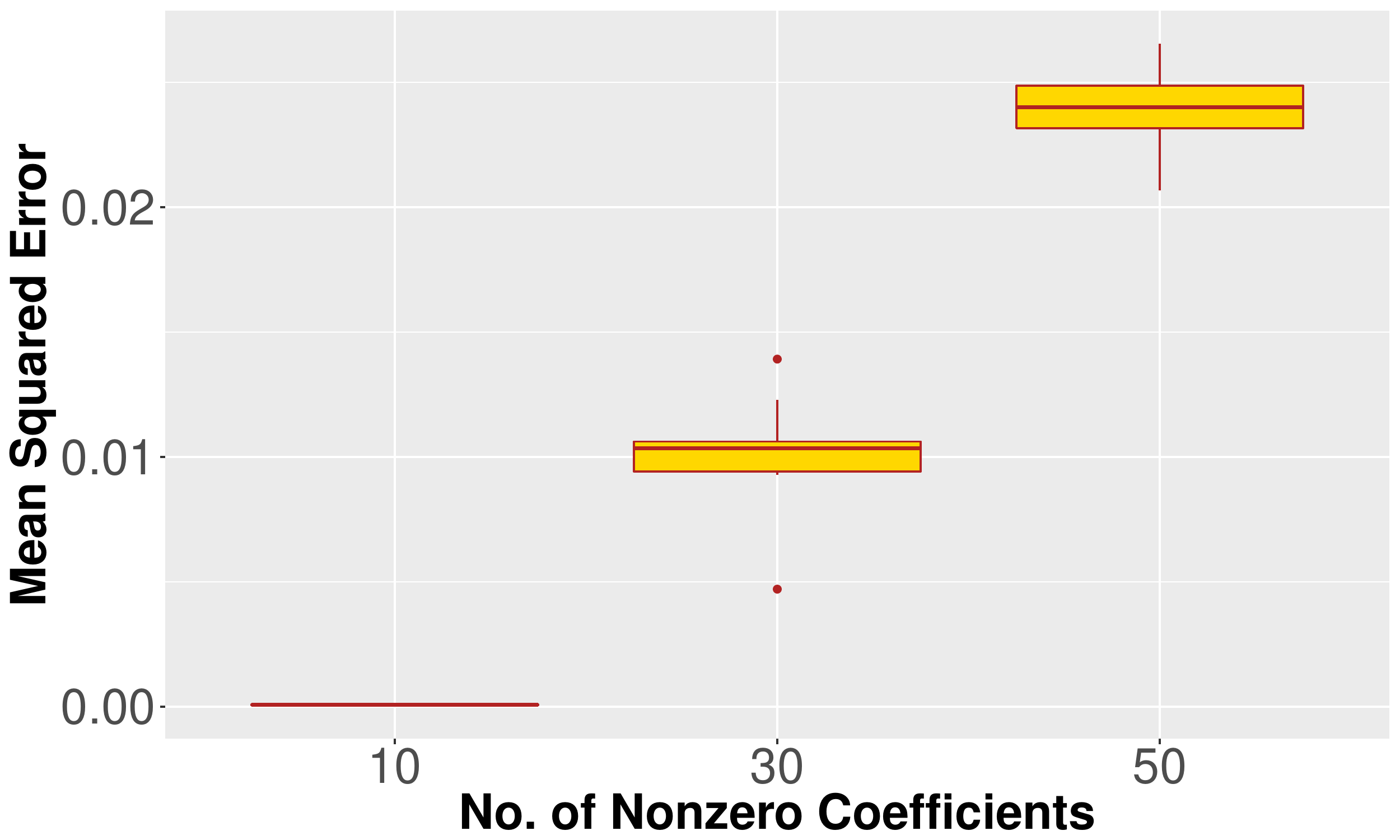}\label{AUC2}}
    \subfigure[MSE of $\bbeta$: MCP]{\includegraphics[width=4.0 cm]{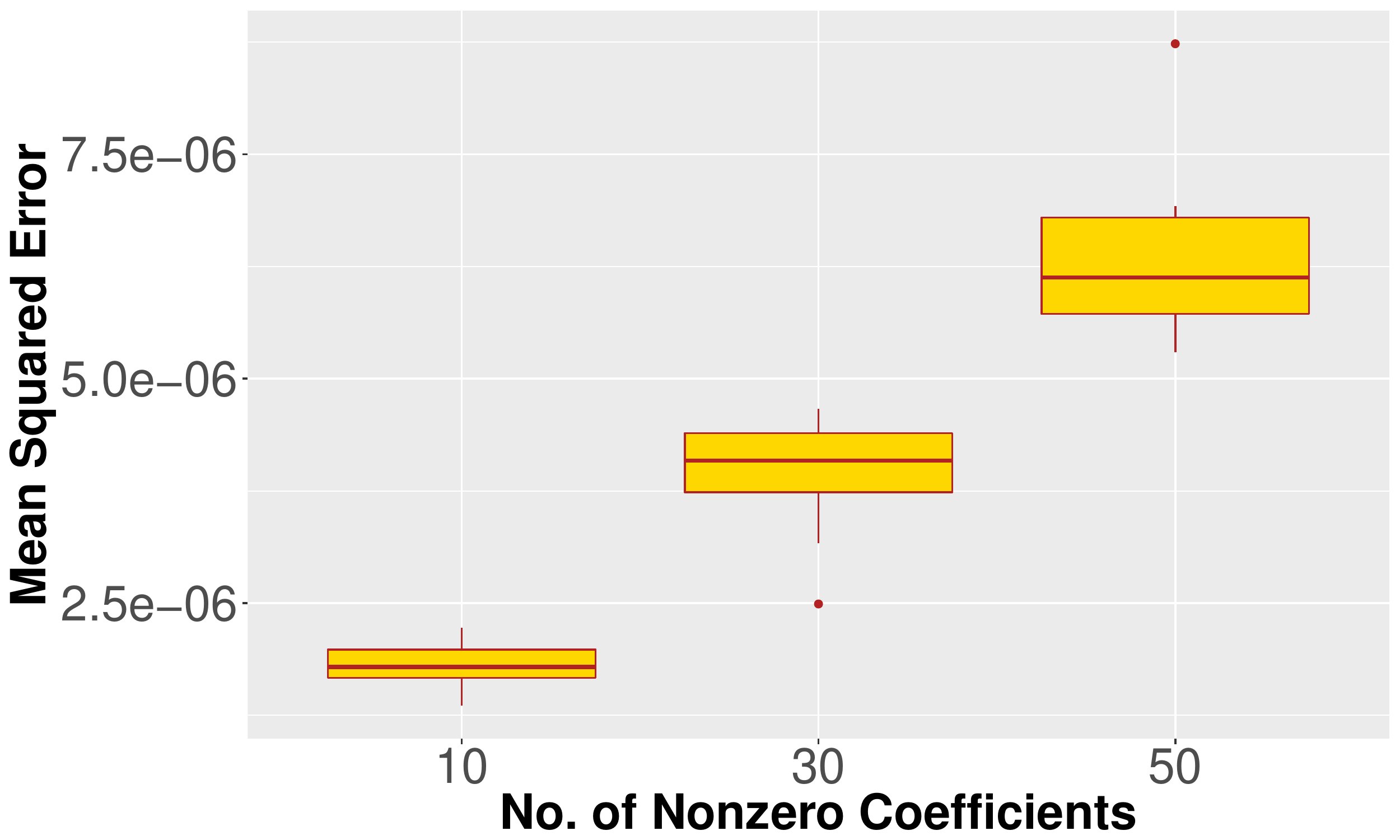}\label{AUC2}}\\
    \subfigure[MSE of nonzero $\bbeta$: CHS, $m=200$]{\includegraphics[width=4.0 cm]{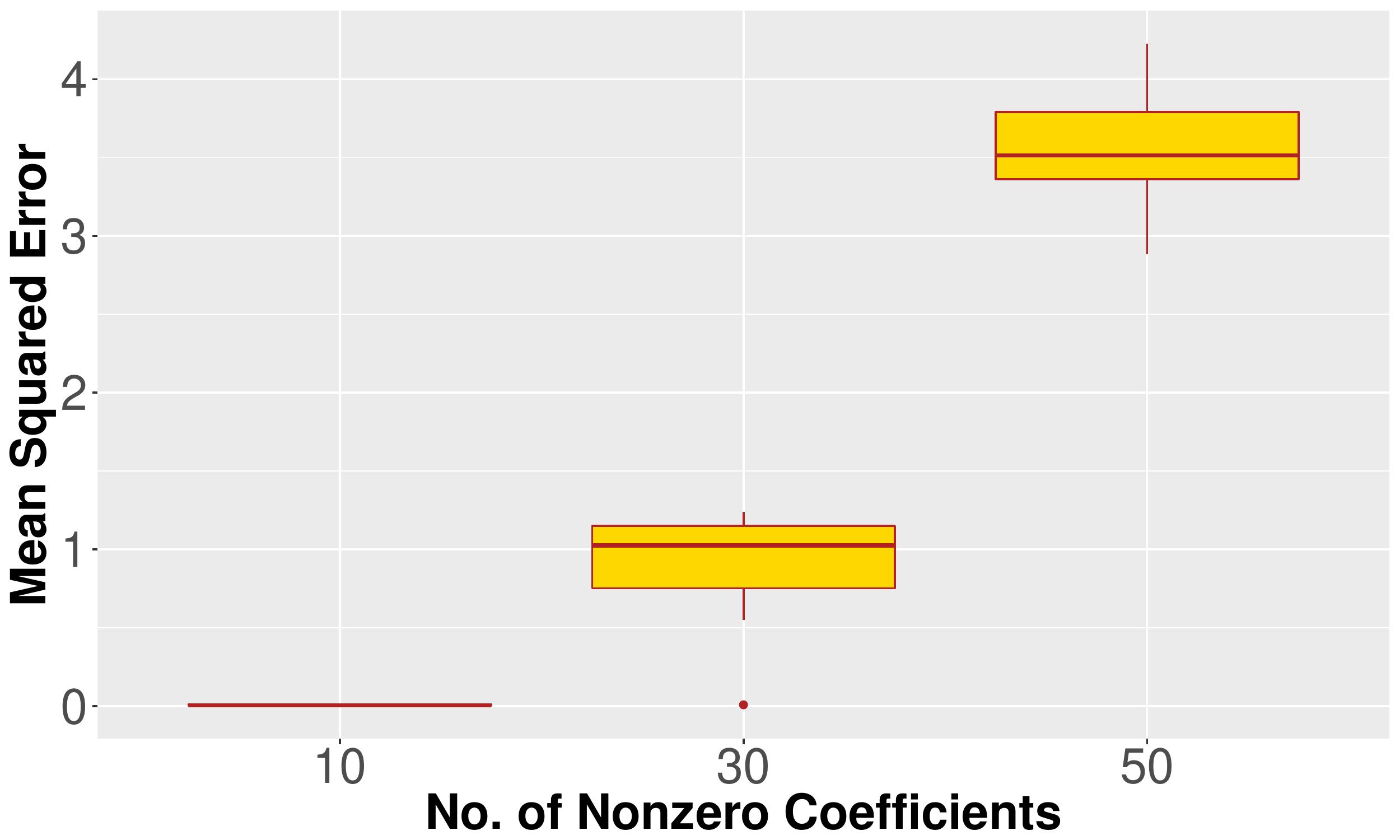}\label{AUC1}}
    \subfigure[MSE of nonzero $\bbeta$: PMCP, $m=200$]{\includegraphics[width=4.0 cm]{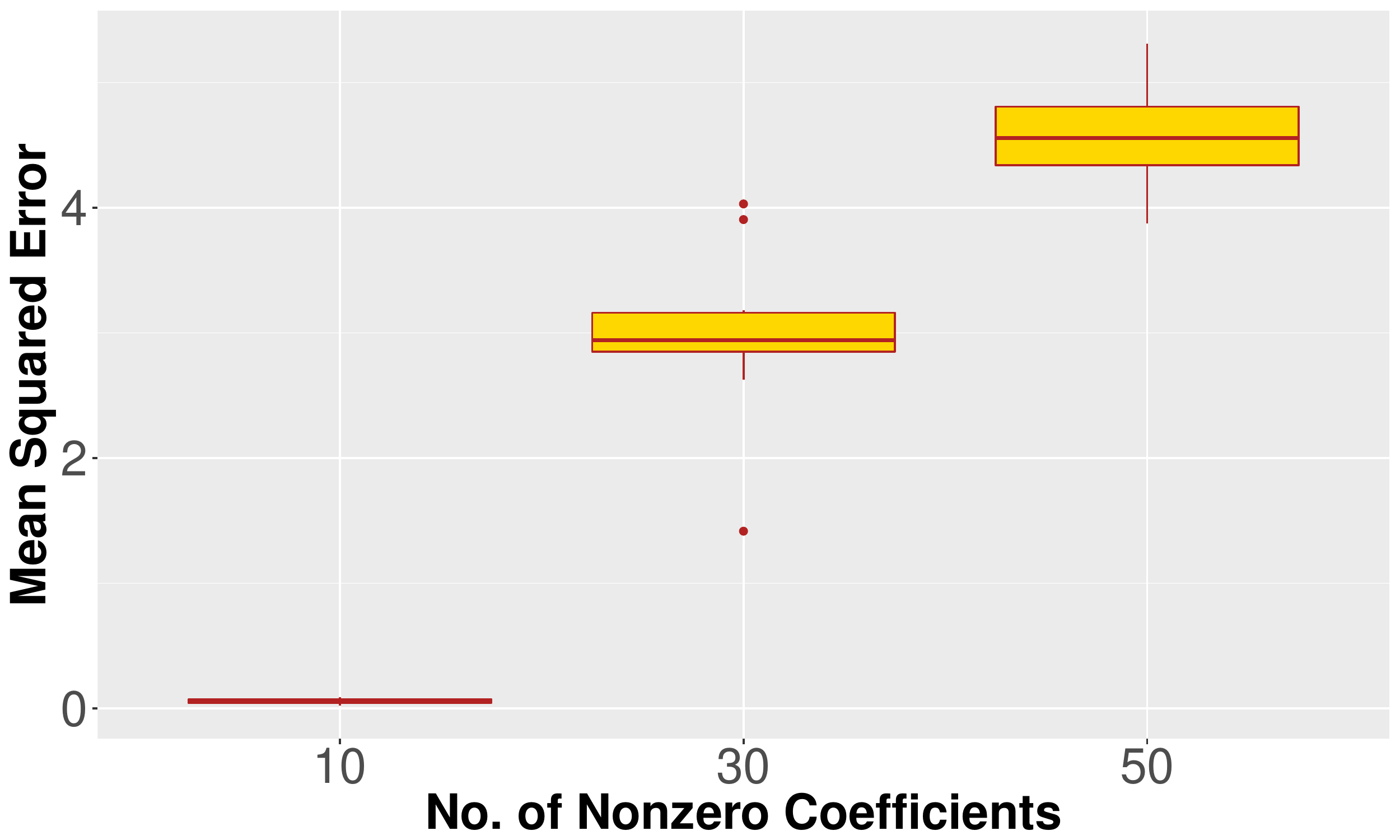}\label{AUC2}}
    \subfigure[MSE of nonzero $\bbeta$: MCP]{\includegraphics[width=4.0 cm]{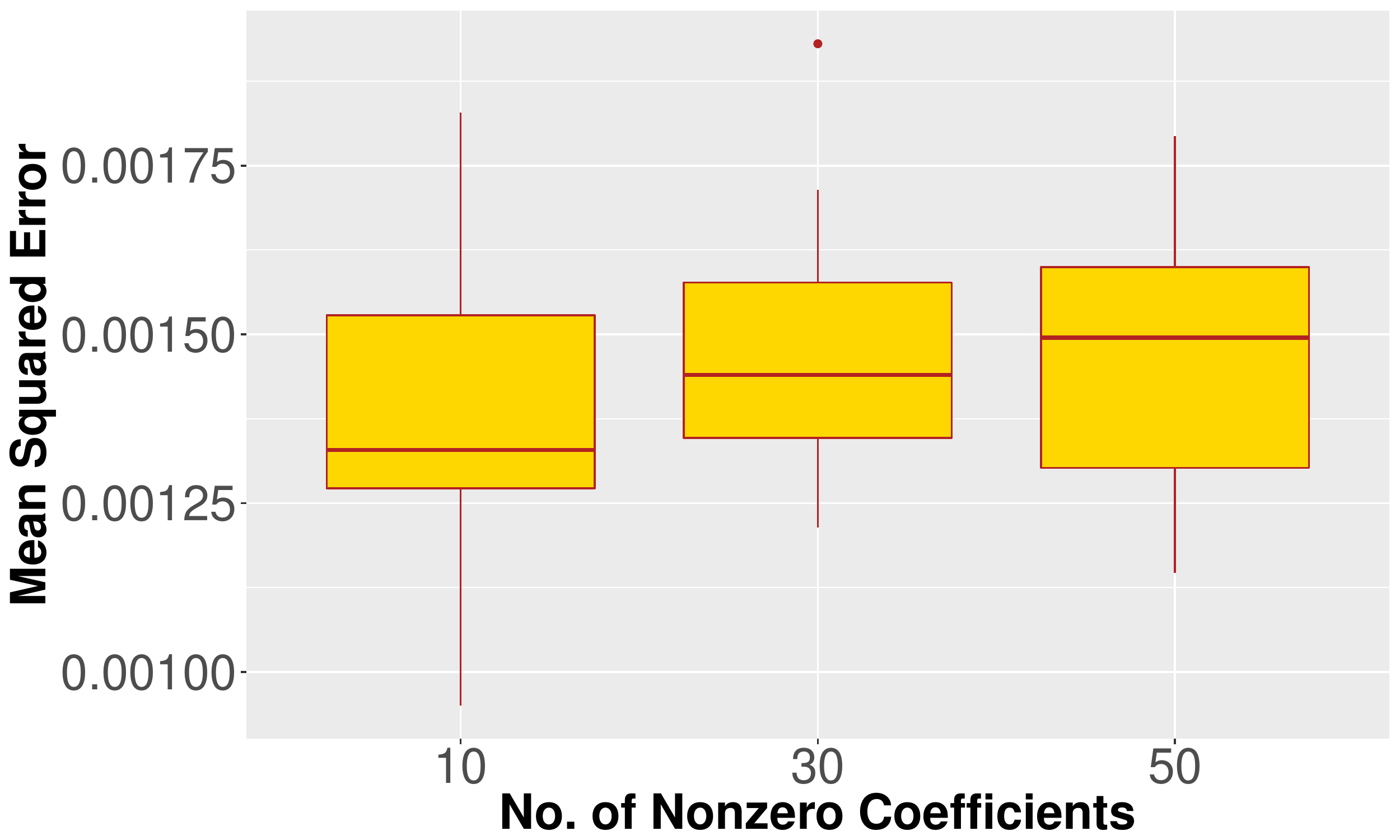}\label{AUC2}}\\
    \subfigure[MSE of $\bbeta$: CHS, $m=400$]{\includegraphics[width=4.0 cm]{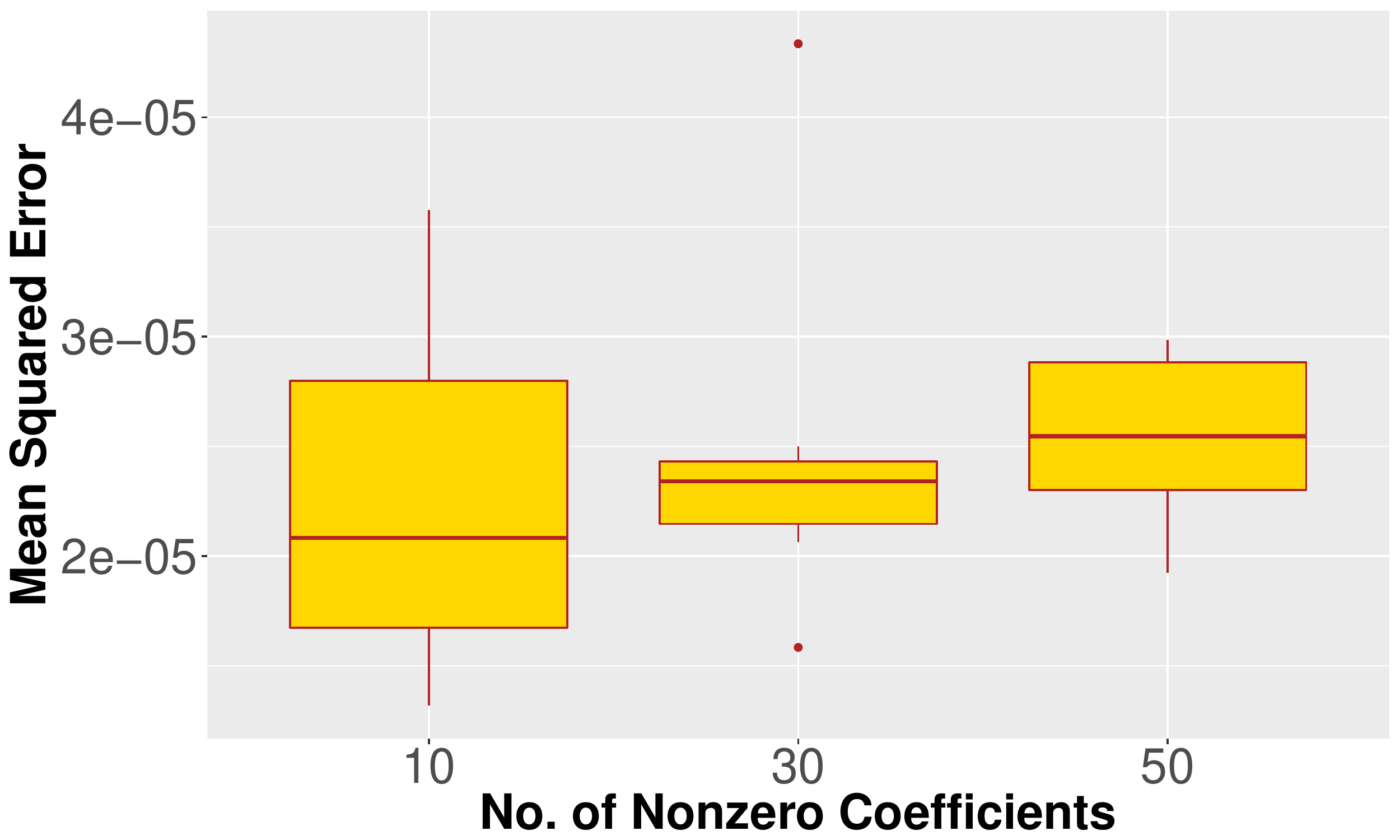}\label{AUC1}}
    \subfigure[MSE of $\bbeta$: PMCP, $m=400$]{\includegraphics[width=4.0 cm]{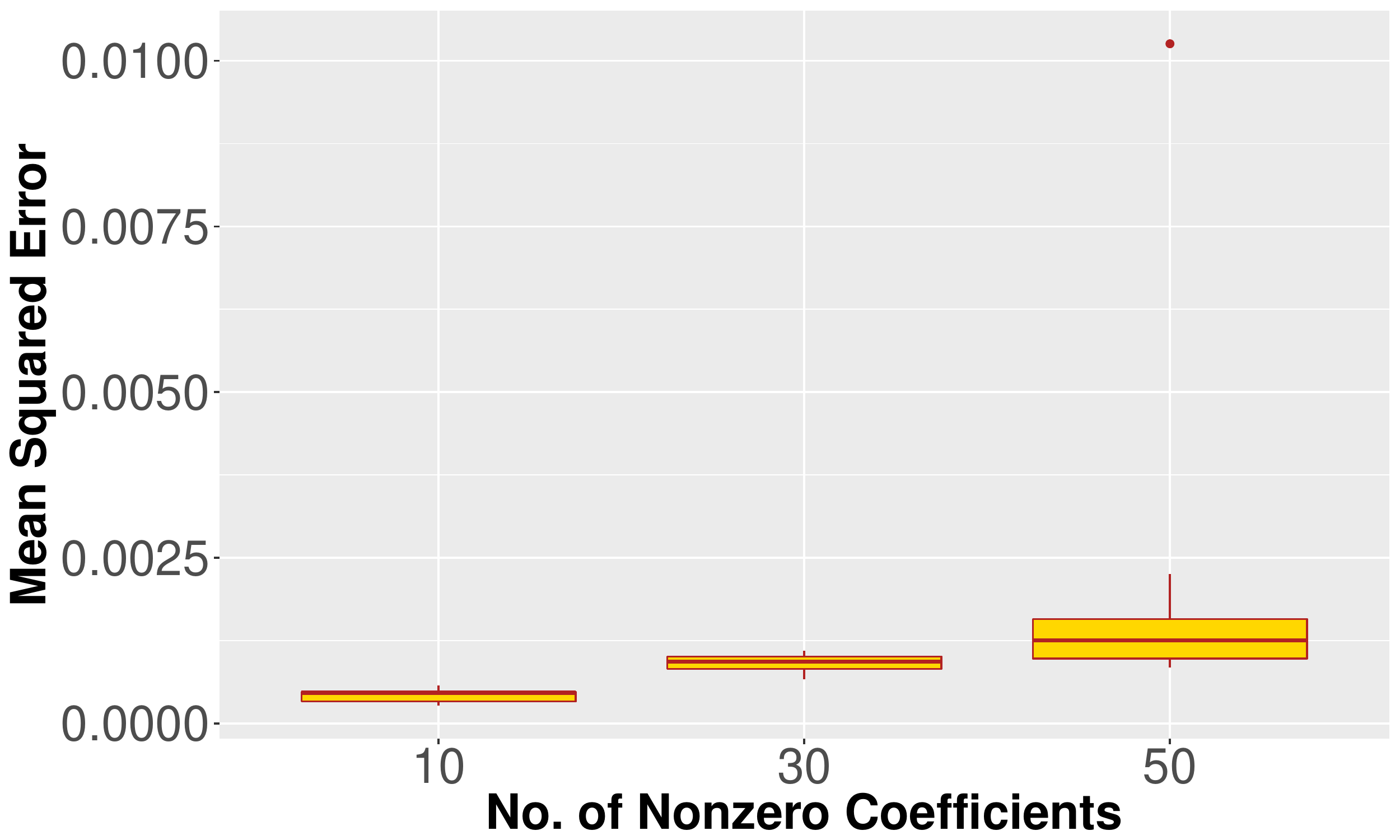}\label{AUC2}}
    \subfigure[MSE of $\bbeta$: MCP]{\includegraphics[width=4.0 cm]{MSE_m400_Lasso_compound.pdf}\label{AUC2}}\\
    \subfigure[MSE of nonzero $\bbeta$: CHS, $m=400$]{\includegraphics[width=4.0 cm]{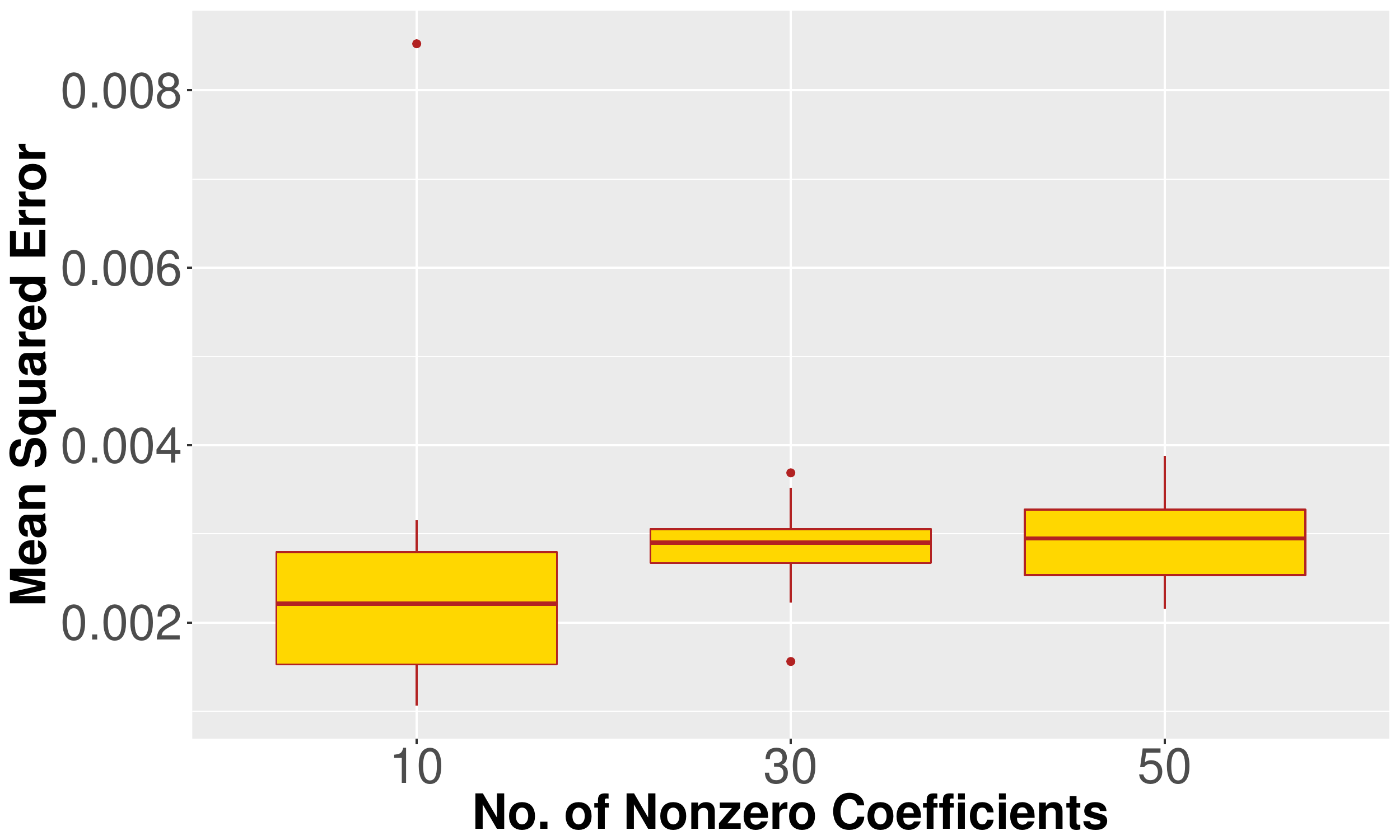}\label{AUC1}}
    \subfigure[MSE of nonzero $\bbeta$: PMCP, $m=400$]{\includegraphics[width=4.0 cm]{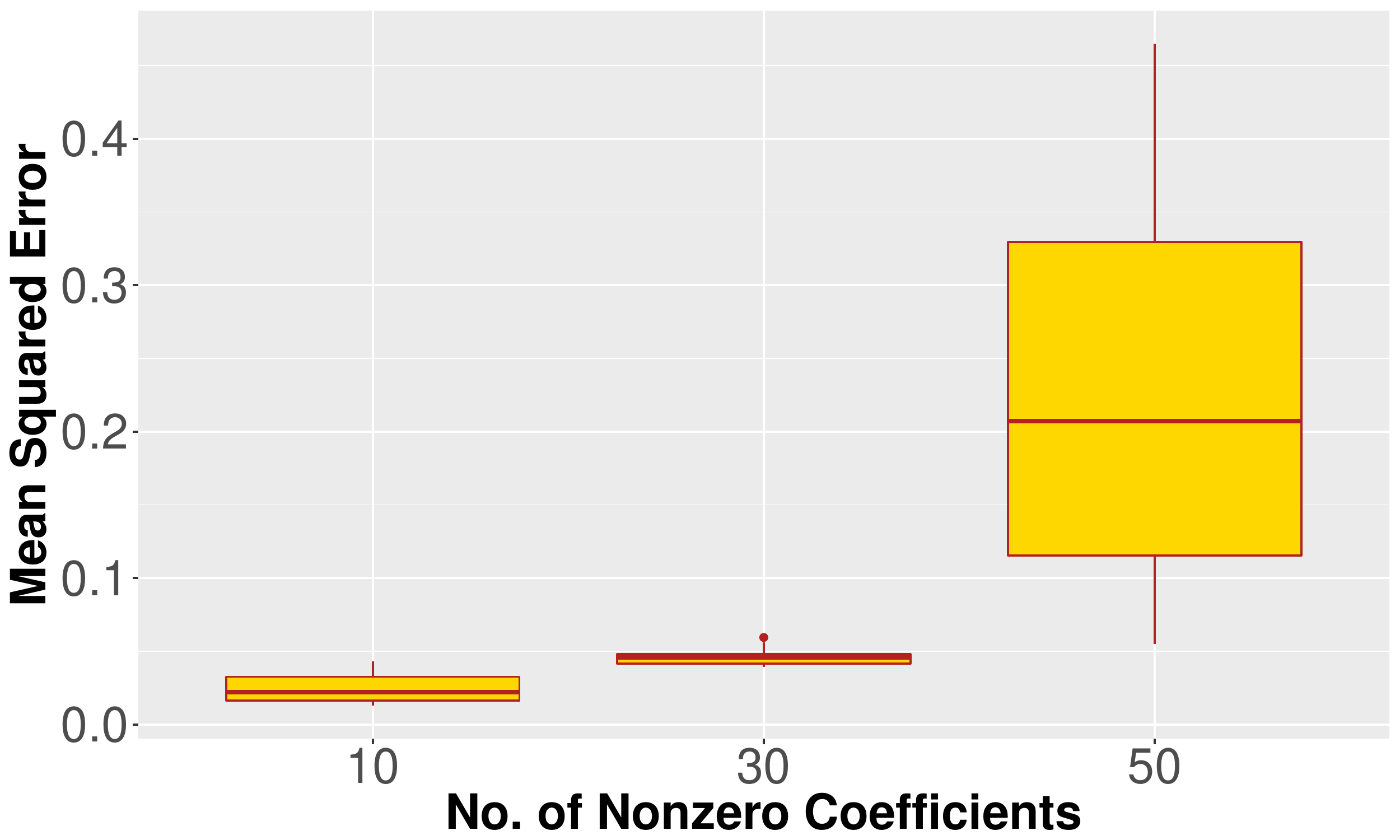}\label{AUC2}}
    \subfigure[MSE of nonzero $\bbeta$: MCP]{\includegraphics[width=4.0 cm]{MSE_nz_m400_Lasso_compound.pdf}\label{AUC2}}
% \end{center}
 \caption{First and third row presenting mean squared error (MSE) of estimating the true predictor coefficient $\bbeta^*$ by a point estimate of $\beta$ from CHS, PMCP and MCP for $m=200$ and $m=400$ respectively. Second and fourth row presenting mean squared error (MSE) of estimating the true nonzero coefficients in $\bbeta^*$ by a point estimate of the corresponding coefficients in $\bbeta$ from CHS, PMCP and MCP for $m=200$ and $m=400$ respectively. All figures correspond to the scenarios where the predictors are generated under the compound correlation structure (Scenario 2). Each figure shows performance of a competitor under the data generated with $10$, $30$ and $50$ nonzero coefficients in $\bbeta^*$. }\label{Fig_MSE_com}
 \end{center}
\end{figure}
While accurate point estimation of $\bbeta^*$ is one of our primary objectives, characterizing uncertainty is of paramount importance given the recent developments in the frequentist literature on characterizing uncertainty in high dimensional regression \citep{javanmard2014confidence,van2014asymptotically,zhang2014confidence}. Although Bayesian procedures provide an automatic characterization of uncertainty, the resulting credible intervals may not possess the correct frequentist coverage in nonparametric/high-dimensional problems \citep{szabo2015frequentist}. To this end, an attractive adaptive property of the shrinkage priors, including horseshoe, is that the length of the intervals automatically adapt between the signal and noise variables, maintaining close to nominal coverage. It is important to see if this property is preserved under data sketching when the horseshoe prior is set on each component of $\bbeta$. Table~\ref{table2} shows that under $m=400$, 95\% credible intervals (CI) of all nonzero coefficients offer closely nominal coverage. While it is also true for $m=200$ and $s=10$, the coverage for nonzero coefficients tend to deteriorate as $s/m$ increases. Comparing the average length of 95\% CIs for all coefficients with the average length of 95\% CIs of nonzero coefficients, we observe that the posterior yields much narrower CIs for coefficients corresponding to the noise predictors. As demonstrated in some of the recent literature \citep{bhattacharya2016fast}, the frequentist procedures of constructing confidence intervals for high dimensional parameters \citep{javanmard2014confidence,van2014asymptotically,zhang2014confidence} in MCP yield approximately equal sized intervals for the signals and noise variables. Additionally, the tuning parameters in the frequentist procedure require substantial tuning to arrive at satisfactory coverage for the noise (though at the cost of under-covering the signals), while our Bayesian approach is naturally auto-tuned.

\begin{table}
\centering
{\small
\caption{Mean squared prediction error$\times 10^3$ for all the competing models under different simulation scenarios. MSPE is computed as $||\bX\hat{\bbeta}-\bX\bbeta^*||^2/n$ for all the competitors.}\label{table3}%
\begin{tabular}
[c]{c|ccc|ccc|ccc|ccc}
\hline
&\multicolumn{3}{c}{Scenario 1, $m=200$} & \multicolumn{3}{c}{Scenario 1, $m=400$} & \multicolumn{3}{c}{Scenario 2, $m=200$} & \multicolumn{3}{c}{Scenario 2, $m=400$}\\
\hline
Sparsity & $10$ & $30$ & $50$ & $10$ & $30$ & $50$ & $10$ & $30$ & $50$ & $10$ & $30$ & $50$ \\
\hline
CHS & 0.62 & 46.56 & 205.67 & 0.51 & 0.57 & 0.61 &  0.53 & 39.22 & 196.78 &  0.47 & 0.59 & 0.64\\
PMCP & 1.95 & 71.97  & 249.70 &  0.62 & 2.28 &  33.19 & 1.36 & 62.89  & 234.63 & 0.58 & 1.75 & 50.49\\
MCP  & 0.02 & 0.07  & 0.10 &  0.02 & 0.07 &  0.10 & 0.03 &  0.07 & 0.12 & 0.03 & 0.07 & 0.12\\
\hline
\end{tabular}

%\begin{tabnote}
%U.S., United States of America; R, respondent.
%\end{tabnote}
}
\end{table}

\begin{table}
\centering
{\small
\caption{Average coverage and average length of 95\% credible intervals of $\beta_j$ for CHS under different simulation cases. Here subscript $nz$ is added when the average coverage and average lengths are calculated for truly nonzero coefficients.}\label{table2}%
\begin{tabular}
[c]{c|ccc|ccc|ccc|ccc}
\hline
&\multicolumn{3}{c}{Scenario 1, $m=200$} & \multicolumn{3}{c}{Scenario 1, $m=400$} & \multicolumn{3}{c}{Scenario 2, $m=200$} & \multicolumn{3}{c}{Scenario 2, $m=400$}\\
\hline
Sparsity & $10$ & $30$ & $50$ & $10$ & $30$ & $50$ & $10$ & $30$ & $50$ & $10$ & $30$ & $50$ \\
\hline
Coverage & 0.99 & 0.99 & 0.98 & 0.99 & 0.99 & 0.98 &  0.99 & 0.99 & 0.98 &  0.99 & 0.99 & 0.98\\
Length & 0.02 & 0.08  & 0.17 &  0.02 & 0.03 &  0.03 & 0.01 &  0.11 & 0.17 & 0.01 & 0.02 & 0.03\\
Coverage$_{nz}$ & 0.97 & 0.86 & 0.68 & 0.95 & 0.97 & 0.97 & 0.98  & 0.89 & 0.63 & 0.95  & 0.96 & 0.95\\
Length$_{nz}$ & 5.72 & 5.93  & 5.19 &  5.53 & 5.90 &  5.83 & 5.49 & 6.59  & 4.42  & 5.51 & 5.79 & 5.77\\
\hline
\end{tabular}
%\begin{tabnote}
%U.S., United States of America; R, respondent.
%\end{tabnote}
}
\end{table}

\section{Real Data Application}\label{sec5}
%\textcolor{blue}{Need to add}.
To illustrate our approach, we present analysis of 
American College of Surgeons National Surgical Quality Improvement Program (ACS NSQIP) data. The ACS NSQIP is a nationally validated, risk-adjusted, outcomes-based program to measure and improve the quality of surgical care. Built by surgeons for surgeons, ACS NSQIP provides participating hospitals with tools, analyses, and reports to make informed decisions about improving quality of care. The data that we focus on consist of information on $n = 2,108$ subjects. For each subject, the variable of interest is days from operation to discharge, which acts as the response variable in our analysis. The remaining $p=1,771$ variables composed of biometric data, surgical codes, post procedure diagnosis along with sex, age, and smoking status interactions are considered as predictors. 

\begin{figure}
  \begin{center}
    \includegraphics[height=7.0 cm]{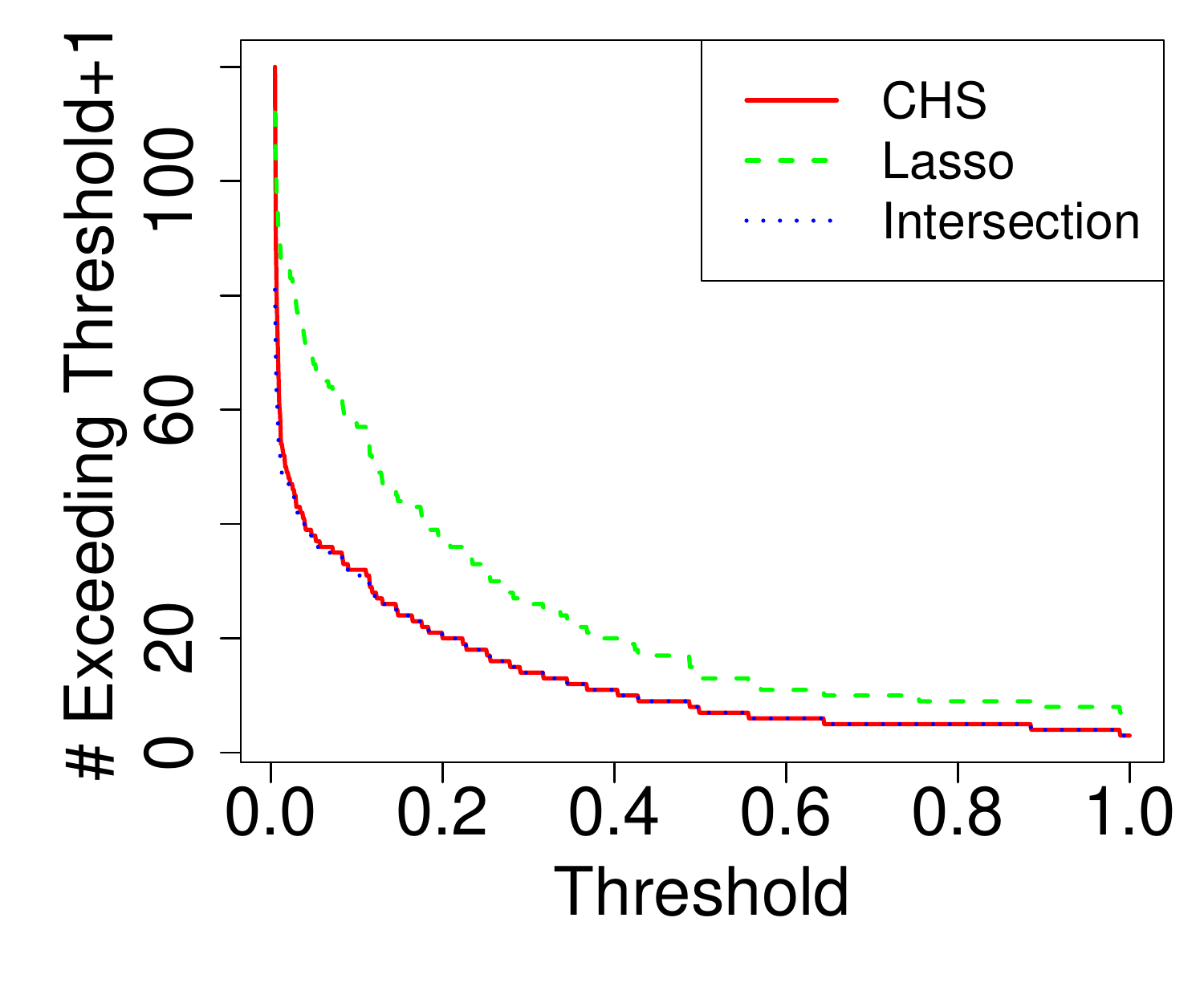}
% \end{center}
 \caption{Plot of the number of variables for which $E[\beta_j|-]$ (from CHS) or $\hat{\beta}_j$ (where $\hat{\beta}_j$ is the lasso estimate of $\beta_j$) exceeds a threshold. The intersecting number of variables from these two methods which exceed the same threshold is also presented. We present the plot for different choices of the threshold ranging between $0.005$ to $1$.
 }\label{Fig_real_data}
 \end{center}
\end{figure}

We fit our compressed horseshoe approach with $m=500$ and compare the results with ordinary uncompressed lasso.
Figure 9 plots the number of entries of predictor coefficients for which the absolute value of the corresponding lasso or compressed horseshoe point estimates exceed a threshold between $0.005$ and $1$.
Also shown is the size of the intersection of these two sets. For smaller thresholds, the number of horseshoe point estimates exceeding the threshold is about the same as Lasso, while for larger thresholds, Lasso identifies more strong signal coefficients than CHS. The size of the intersection closely tracks the minimum size of the two sets, suggesting that coefficients are similarly ordered in Lasso and CHS, the only difference being Lasso estimating little higher magnitude for the coefficients. Finally, Lasso provides no notion of uncertainty in the selected variables such as that conveyed by the posterior marginals of the compressed horseshoe.

\section{Conclusion}\label{sec6}
This article presents a data sketching/compression approach in high dimensional linear regression with Gaussian scale mixture priors. The proposed approach ensures privacy of
the original data by revealing little information about it to the analyst. Additionally, it leads to a
massive reduction in computation for big $n$ and $p$. Simulation studies show advantage of data compression over naive sub-sampling of data, as well as
competitive performance of the approach with uncompressed data, especially in presence of a high degree of sparsity. Asymptotic results throw light on the interplay of
sparsity, dimension of the compression matrix, sample size and the number of features.

Although our approach is applied to the Horseshoe prior, it lends easy usage to any other Gaussian scale mixture prior, such as the Generalized Double Pareto \citep{armagan2013generalized} or the normal gamma prior \citep{griffin2010inference}. The data sketching approach also finds natural extension to high dimensional
binary or categorical regression using the data augmentation approach. While simulation studies show promising empirical performance of such an approach, we plan to put forth effort to develop theoretical results in a similar spirit as Section~\ref{sec3}. We also plan to extend the data sketching approach to high dimensional nonparametric models with big $n$ and $p$.

\section{Acknowledgement}
The research of Rajarshi Guhaniyogi is partially supported by grants from the Office of Naval Research (ONR-BAA N000141812741) and the National Science Foundation (DMS-1854662).

%\appendix

%\appendixone
\section*{Appendix}
We begin by stating an important result from the random matrix theory, the proof of which is immediate following Theorem 5.31 and Corollary
5.35 of \cite{vershynin2010introduction}.
\begin{lemma}\label{lem1}
Consider the $m_n\times n$ compression matrix $\bPhi_n$ with each entry being drawn independently from $N(0,1/n)$. Then, almost surely
\begin{align}\label{compression}
\frac{(\sqrt{n}-\sqrt{m_n}-o(\sqrt{n}))^2}{n}\leq e_{min}(\bPhi_n\bPhi_n')\leq  e_{max}(\bPhi_n\bPhi_n')\leq \frac{(\sqrt{n}+\sqrt{m_n}+o(\sqrt{n}))^2}{n},
\end{align}
when both $m_n,n\rightarrow\infty$.
\end{lemma}

\begin{lemma}\label{lem2}
Let $P_{\bbeta_n^*}$ denotes the probability distribution of $\by_n$,
\begin{align}\label{eq:lem1}
\mathcal{H}_n=\left\{\by_n:\int \frac{f(\tilde{\by}_n|\bbeta_n)}{f(\tilde{\by}_n|\bbeta_n^*)} \pi(\bbeta_n)d\bbeta_n\leq \exp(-C_1m_n\epsilon_n^2)\right\},\:\:\mbox{for a constant}\:C_1>0.
\end{align}
Then $P_{\bbeta_n^*}(\mathcal{H}_n)\rightarrow 0$ as $n\rightarrow\infty$.
\end{lemma}
\begin{proof}
Denote $\tilde{\by}_n=\bPhi_n \by_n$ and $\tilde{\bX}_n=\bPhi_n \bX_n$. For two densities $g_1,g_2$, denote $K(g_1,g_2)=\int g_1\log(g_1/g_2)$ and 
$V(g_1,g_2)=\int g_1((\log(g_1/g_2))-K(g_1,g_2))^2$. Define
\begin{align}
\mathcal{H}_{1n}=\left\{K(f(\tilde{\by}_n|\bbeta_n^*),f(\tilde{\by}_n|\bbeta_n))\leq m_n\epsilon_n^2,\:\:V(f(\tilde{\by}_n|\bbeta_n^*),f(\tilde{\by}_n|\bbeta_n))\leq m_n\epsilon_n^2\right\}.
\end{align}
By Lemma 10 of \cite{ghosal2007convergence}, to show (\ref{eq:lem1}) it is enough to show that $\Pi(\mathcal{H}_{1n})=O(\exp(-C_2m_n\epsilon_n^2))$, for some constant $C_2>0$.
Let $e_{k,n}$, $1\leq k\leq m_n$ be the ordered eigenvalues of $(\bPhi_n\bPhi_n')^{-1}$. Then with little algebra, we obtain
\begin{align}\label{eq2}
K(f(\tilde{\by}_n|\bbeta_n^*),f(\tilde{\by}_n|\bbeta_n))&=\frac{1}{2}\left\{\sum_{k=1}^{m_n}(e_{k,n}-1-\log(e_{k,n}))+||\bPhi_n\bX_n(\bbeta_n-\bbeta_n^*)||^2\right\}\nonumber\\
V(f(\tilde{\by}_n|\bbeta_n^*),f(\tilde{\by}_n|\bbeta_n))&=\sum_{k=1}^{m_n}\frac{(1-e_{k,n})^2}{2}+||(\bPhi_n\bPhi_n')^{-1}(\bPhi_n\bX_n(\bbeta_n-\bbeta_n^*)||^2.
\end{align}
Expanding $\log(e_{k,n})$ in the powers of $(1-e_{k,n})$ and using Lemma 1 of \cite{jeong2020unified},  $(e_{k,n}-1-\log(e_{k,n}))-(1-e_{k,n})^2/2\rightarrow 0$, as $n\rightarrow\infty$. Another use of Lemma 1 of
\cite{jeong2020unified} yields, $\sum_{k=1}^{m_n}(1-e_{k,n})^2\leq\tilde{C}_1 ||\bI-\bPhi_n\bPhi_n'||_F^2\leq\tilde{C}_2 m_n/n\leq m_n\theta_n^2$, for some constants $\tilde{C}_1,\tilde{C}_2>0$, and for all large $n$. Using Lemma~\ref{lem1},
$e_{k,n}\rightarrow 1$, as $n\rightarrow\infty$, for all $k=1,...,m_n$. Thus, from (\ref{eq2}), for some constant $C_3>0$,
\begin{align}\label{prob_KL}
\Pi(\mathcal{H}_{1n})&\geq\Pi\left(\left\{\bbeta_n: ||\bPhi_n\bX_n(\bbeta_n-\bbeta_n^*)||^2\leq C_3 m_n\theta_n^2\right\}\right)\nonumber\\
&=\Pi\left(\bbeta_n: (\bbeta_n-\bbeta_n^*)'\bX_n'\bPhi_n'\bPhi_n\bX_n(\bbeta_n-\bbeta_n^*)\leq C_3 m_n\theta_n^2\right)\nonumber\\
&=\Pi\left(\bbeta_n: (\bbeta_n-\bbeta_n^*)'(\tilde{\bX}_n'\tilde{\bX}_n/m_n)(\bbeta_n-\bbeta_n^*)\leq C_3\theta_n^2\right)\nonumber\\
&\geq\Pi\left(\bbeta_n: (\bbeta_n-\bbeta_n^*)'(\bbeta_n-\bbeta_n^*)\leq\eta^{-1}C_3\theta_n^2\right)\nonumber\\
&\geq \Pi\left(\bbeta_n: ||\bbeta_n-\bbeta_n^*||_1\leq 2\tilde{C}_3\theta_n\right),
%&\geq\Pi\left(\left\{\bbeta_n: E_{\mathcal{S}}E_{X}\left[||\Phi X_DB(\gamma-\gamma_0)||^2\right]+E_{\mathcal{S}}E_{X}\left[||\Phi X_D\eta_0||^2\right]\lesssim m_n\epsilon_n^2/2\right\}\right),
\end{align}
where the inequality in the fourth line follows from the fact that there exists $\tilde{\eta}>0$ such that $||\tilde{\bX}_n\bv||_2^2\leq \tilde{\eta}||\bX_n v||_2^2$, for all $\bv$ \citep{ahfock2017statistical}. This implies that 
$e_{max}(\tilde{\bX}_n'\tilde{\bX}_n)=\sup\limits_{||\bv||_2=1}\bv'\tilde{\bX}_n'\tilde{\bX}_n\bv\leq\tilde{\eta} \sup\limits_{||\bv||_2=1}\bv'\bX_n'\bX_n\bv=\tilde{\eta}e_{max}(\bX_n'\bX_n).$ The inequality then follows by Assumption (A) that ensures 
$e_{max}(\bX_n'\bX_n)\leq np_n$.

Now,
$\{\bbeta_n:||\bbeta_n-\bbeta_n^*||_1<2\tilde{C}_3\theta_n\}\supset \{|\beta_{j,n}|\leq \tilde{C}_3\theta_n/p_n,\:\forall\:j\notin\bxi^*\}\cap\{|\beta_{j,n}-\beta_{j,n}^*|\leq \tilde{C}_3\theta_n/s_n\:\forall\:j\in\bxi^*\}$.
Now, $\Pi(|\beta_{j,n}|\leq \tilde{C}_3\theta_n/p_n,\:\forall\:j\notin\bxi^*)\geq \prod_{j\notin\bxi^*}\Pi(|\beta_{j,n}|\leq a_n)\geq (1-p_n^{-1-u})^{p_n}\rightarrow 1$, as $n\rightarrow \infty$. Here the first inequality follows as $a_n=\sqrt{s_n\log(p_n)/m_n}/p_n$ and $s_n\log(p_n)/m_n\rightarrow 0$. The second inequality follows by Assumption (F).
On the other hand, $\Pi(|\beta_{j,n}-\beta_{j,n}^*|\leq \tilde{C}_3\theta_n/s_n,\:\forall\:j\in\bxi^*)\geq (2\tilde{C}_3\theta_n/s_n\inf\limits_{[-M_n,M_n]}h_{\mu_n}(x))^{s_n}$, which holds for all large $n$ as $|\beta_{j,n}^*|<M_n/2$ and $\tilde{C}_3\theta_n/s_n\rightarrow 0$ as $n\rightarrow\infty$. Thus,
$-\log(\Pi(|\beta_{j,n}-\beta_{j,n}^*|\leq \tilde{C}_3\theta_n/s_n,\:\forall\:j\in\bxi^*))\leq O(s_n\log(p_n))=O(m_n\theta_n^2)$, by Assumptions (C) and (G). This proves our result.
\end{proof}

\noindent\textbf{Proof of Theorem~\ref{thm_main}}\\
\begin{proof}
Denote $\tilde{\by}_n=\bPhi_n \by_n$ and $\tilde{\bX}_n=\bPhi_n \bX_n$ and consider the following conditions,
\begin{enumerate}
\item \textbf{Condition (i):} $\exists$ a test function $\kappa_n$ s.t.
\begin{align*}
E_{\bbeta_n^*}(\kappa_n)\leq \exp(-\tilde{c}_3m_n\theta_n^2),\:\:\:\sup\limits_{\bbeta\in\mathcal{C}_n}E_{\bbeta_n}(1-\kappa_n)\leq \exp(-\tilde{c}_4m_n\theta_n^2),
\end{align*}
for some constants $\tilde{c}_3,\tilde{c}_4>0$, respectively.
%\item \textbf{Condition (ii):} $\Pi(\mathcal{B}_n)\leq \exp(-\tilde{c}_5m_n\theta_n^2)$, for some constant $\tilde{c}_5>0$.
\item \textbf{Condition (ii):} For $\mathcal{H}_n=\left\{\by_n:\int \frac{f(\tilde{\by}_n|\bbeta_n)}{f(\tilde{\by}_n|\bbeta_n^*)} \pi(\bbeta_n)d\bbeta_n\geq \exp(-\tilde{c}_6m_n\epsilon_n^2)\right\}$, $P_{\bbeta_n^*}(\mathcal{H}_n)\rightarrow 1$, as $n\rightarrow\infty$, for some $0<\tilde{c}_6<\tilde{c}_4$.
    %$P_{\bbeta_n^*}\left(\frac{\int f(\tilde{\by}_n|\bbeta_n)\pi(\bbeta_n)d\bbeta_n}{f(\tilde{\by}_n|\bbeta_n^*)}\geq \exp(-\tilde{c}_6m_n\theta_n^2)\right)\geq 1-\exp(-\tilde{c}_7m_n\theta_n^2)$, for $0<\tilde{c}_6<\min(\tilde{c}_4,\tilde{c}_5)$ and for some $\tilde{c}_7>0$.
\end{enumerate}
We begin by showing that Conditions (i)-(ii) are sufficient to prove $E_{\bbeta_n^*}\Pi(\mathcal{C}_n)\rightarrow 0$, as $m_n,n\rightarrow\infty$. 
Note that
\begin{align}\label{eq_initial}
E_{\bbeta_n^*}\Pi(\mathcal{C}_n)&\leq E_{\bbeta_n^*}[\kappa_n]+E_{\bbeta_n^*}\left[\frac{(1-\kappa_n)\int_{\mathcal{C}_n}\frac{f(\tilde{\by}_n|\bbeta_n)}{f(\tilde{\by}_n|\bbeta_n^*)}\pi_n(\bbeta_n)d\bbeta_n}{\int \frac{f(\tilde{\by}_n|\bbeta_n)}{f(\tilde{\by}_n|\bbeta_n^*)}\pi_n(\bbeta_n)d\bbeta_n}1_{\by_n\in\mathcal{H}_n}\right]+P_{\bbeta_n^*}(\mathcal{H}_n^c)\nonumber\\
&\leq E_{\bbeta_n^*}[\kappa_n]+\sup\limits_{\bbeta_n\in\mathcal{C}_n}E_{\bbeta_n}[(1-\kappa_n)]\Pi(\mathcal{C}_n)\exp(\tilde{c}_6m_n\theta_n^2)+P_{\bbeta_n^*}(\mathcal{H}_n^c)\nonumber\\
&\leq E_{\bbeta_n^*}[\kappa_n]+\sup\limits_{\bbeta_n\in\mathcal{C}_n}E_{\bbeta_n}[(1-\kappa_n)]\exp(\tilde{c}_6m_n\theta_n^2)+P_{\bbeta_n^*}(\mathcal{H}_n^c),
\end{align}
where the inequality in the second line follows from Condition (ii). Condition (i) can now be applied to show that $E_{\bbeta_n^*}\Pi(\mathcal{C}_n)\rightarrow 0$, as $n\rightarrow \infty$.

It remains to prove Conditions (i)-(ii) which we prove below.\\
\underline{\textbf{Proof of Condition (i):}}\\
Define a sequence of test functions\\ $\kappa_n=\max_{\bxi\supset\bxi^*,|\bxi|\leq s_n+\tilde{s}_n}1\{||(\tilde{\bX}_{n,\bxi}'\tilde{\bX}_{n,\bxi})^{-1}\tilde{\bX}_{n,\bxi}'\tilde{\by}_n-\bbeta_{n,\bxi}^*||_2\geq \theta_n\}$, where $\tilde{s}_n$ is defined later. Let $\hat{\bbeta}_{n,\bxi}=(\tilde{\bX}_{n,\bxi}'\tilde{\bX}_{n,\bxi})^{-1}\tilde{\bX}_{n,\bxi}'\tilde{\by}_n$. Then
\begin{align*}
& E_{\bbeta_n^*}(\kappa_n)\leq \sum_{\bxi\supset\bxi^*,|\bxi|\leq s_n+\tilde{s}_n}P_{\bbeta_n^*}(||\hat{\bbeta}_{n,\bxi}-\bbeta_{n,\bxi}^*||_2\geq \theta_n)\\
&=\sum_{\bxi\supset\bxi^*,|\bxi|\leq s_n+\tilde{s}_n}P_{\bbeta_n^*}((\hat{\bbeta}_{n,\bxi}-\bbeta_{n,\bxi}^*)'(\hat{\bbeta}_{n,\bxi}-\bbeta_{n,\bxi}^*)\geq \theta_n^2)\\
&\leq \sum_{\xi\supset\bxi^*,|\bxi|\leq s_n+\tilde{s}_n}P_{\bbeta_n^*}\left((\hat{\bbeta}_{n,\bxi}-\bbeta_{n,\bxi}^*)'\tilde{\bX}_{n,\bxi}'(\bPhi_n\bPhi_n')^{-1}\tilde{\bX}_{n,\bxi}(\hat{\bbeta}_{n,\bxi}-\bbeta_{n,\bxi}^*)\geq \frac{\tilde{C}_4\theta_n^2m_nn}{(\sqrt{m_n}+\sqrt{n}+o(\sqrt{n}))^2}\right)\\
& \leq \sum_{\bxi\supset\bxi^*,|\bxi|\leq s_n+\tilde{s}_n}P_{\bbeta_n^*}\left(\chi_{|\bxi|}^2\geq \frac{\tilde{C}_4\theta_n^2m_nn}{(\sqrt{m_n}+\sqrt{n}+o(\sqrt{n}))^2}\right)\\
& \leq \sum_{\bxi\supset\bxi^*,|\bxi|\leq s_n+\tilde{s}_n}P_{\bbeta_n^*}(\chi_{|\bxi|}^2\geq (1-\delta)\tilde{C}_4\theta_n^2m_n)
\leq {p_n \choose \tilde{s}_n+s_n}\exp(-2\tilde{c}_3\theta_n^2m_n)\leq \exp(-\tilde{c}_3\theta_n^2m_n),
\end{align*}
for some constant $\tilde{c}_3>0$, where $\chi_{|\bxi|}^2$ is a $\chi^2$ random variable with $|\bxi|$ degrees of freedom.
Here the inequality in the third line follows from two results. First, by Lemma~\ref{lem1}, $e_{min}((\bPhi_n\bPhi_n')^{-1})\geq n/(\sqrt{n}+\sqrt{m_n}+o(\sqrt{n}))^2$ almost surely. Second, $e_{min}(\tilde{\bX}_{n,\bxi}'\tilde{\bX}_{n,\bxi}/m_n)\geq e_{min}(\bX_{n,\bxi}'\bX_{n,\bxi}/n) \tilde{\eta}$, for some $\tilde{\eta}>0$, by \cite{ahfock2017statistical}. Thus, using Assumption (D), it follows that $e_{min}(\tilde{\bX}_{n,\bxi}'\tilde{\bX}_{n,\bxi}/m_n)\geq \tilde{C}_4$, for some constant $\tilde{C}_4>0$ and for all $\bxi\supset\bxi^*$ such that $|\bxi|\leq s_n+\tilde{s}_n$. The first inequality in the fifth line follows due to the fact that
$n/(\sqrt{m_n}+\sqrt{n}+o(\sqrt{n}))^2\rightarrow 1$ as $n\rightarrow\infty$. Hence $n/(\sqrt{m_n}+\sqrt{n}+o(\sqrt{n}))^2\geq 1-\delta$ for some $\delta\in (0,1)$, for all large $n$. The second inequality in the fifth line in obtained by applying the Bernstein inequality \citep{song2017nearly}. To accomplish the third inequality in the fifth line, we set $\tilde{s}_n=\frac{\tilde{c}_3m_n\theta_n^2}{2\log(p_n)}.$ Such an $\tilde{s}_n$ exists since $s_n=o(\bar{s}_n\log(p_n))$, by assumption (C). The inequality is then obtained by the fact that ${p_n \choose \tilde{s}_n+s_n}\leq p_n^{\tilde{s}_n+s_n}\leq \exp((\tilde{s}_n+s_n)\log(p_n))\leq \exp(\tilde{c}_3m_n\theta_n^2)$, using assumptions (C).

Consider $\bzeta=\bxi^*\cup\{j:|\beta_{j,n}|\geq a_n\}$. Then $\bzeta\in\{\bxi: \bxi\supset\bxi^*,|\bxi|\leq s_n+\tilde{s}_n\}$. Then
\begin{align*}
\sup\limits_{\bbeta_n\in\mathcal{C}_n}E_{\bbeta_n}(1-\kappa_n)\leq \sup\limits_{\bbeta_n\in\mathcal{C}_n}\{1-P_{\bbeta_n}(||\hat{\bbeta}_{n,\bzeta}-\bbeta_{n,\bzeta}^*||_2\geq \theta_n)\}
=\sup\limits_{\bbeta_n\in\mathcal{C}_n}P_{\bbeta_n}(||\hat{\bbeta}_{n,\bzeta}-\bbeta_{n,\bzeta}^*||_2\leq \theta_n).
\end{align*}
Under $\mathcal{C}_n$, $||\bbeta_{n,\bzeta}-\bbeta_{n,\bzeta}^*||_2\geq ||\bbeta_{n}-\bbeta_{n}^*||_2-||\bbeta_{n,\bzeta^c}-\bbeta_{n,\bzeta^c}^*||_2
\geq 3\theta_n-a_np_n\geq 2\theta_n$. Here the last inequality follows due to the fact that $\bbeta_{n,\bzeta^c}^*=0$ and for any $j\in\bzeta^c$, $|\beta_{n,j}|\leq a_n$ and $a_n=\theta_n/p_n$. Using the above fact, we have
\begin{align*}
&\sup\limits_{\bbeta_n\in\mathcal{C}_n}P_{\bbeta_n}(||\hat{\bbeta}_{n,\bzeta}-\bbeta_{n,\bzeta}^*||_2\leq \theta_n)
\leq\sup\limits_{\bbeta_n\in\mathcal{C}_n} P_{\bbeta_n}(||\hat{\bbeta}_{n,\bzeta}-\bbeta_{n,\bzeta}||_2\geq ||\bbeta_{n,\bzeta}-\bbeta_{n,\bzeta}^*||_2-\theta_n)\\
&=\sup\limits_{\bbeta_n\in\mathcal{C}_n} P_{\bbeta_n}(||\hat{\bbeta}_{n,\bzeta}-\bbeta_{n,\bzeta}||_2\geq\theta_n)\\
&\leq \sup\limits_{\bbeta_n\in\mathcal{C}_n}P_{\bbeta_n}\left((\hat{\bbeta}_{n,\bzeta}-\bbeta_{n,\bzeta})'\tilde{\bX}_{n,\bzeta}'(\bPhi_n\bPhi_n')^{-1}\tilde{\bX}_{n,\bzeta}(\hat{\bbeta}_{n,\bzeta}-\bbeta_{n,\bzeta})\geq \frac{\tilde{C}_4\theta_n^2m_nn}{(\sqrt{m_n}+\sqrt{n}+o(\sqrt{n}))^2}\right)\\
& \leq \sup\limits_{\bbeta_n\in\mathcal{C}_n}P_{\bbeta_n}\left(\chi_{|\bzeta|}^2\geq \frac{\tilde{C}_4\theta_n^2m_n n}{(\sqrt{m_n}+\sqrt{n}+o(\sqrt{n}))^2}\right)
 \leq \sup\limits_{\bbeta_n\in\mathcal{C}_n}P_{\bbeta_n}(\chi_{|\bzeta|}^2\geq (1-\delta)\tilde{C}_4\theta_n^2m_n)\\
&\leq \exp(-\tilde{c}_4m_n\theta_n^2),\:\mbox{for some constant $\tilde{c}_4>0$.}
\end{align*}

\noindent\underline{\textbf{Proof of Condition (ii):}}\\
This is proved using Lemma~\ref{lem2} by suitably choosing $C_1>0$ in the statement of the lemma to be less than $\tilde{c}_4$.
\end{proof}

\noindent\textbf{Proof of Theorem~\ref{consistency_new}}
\begin{proof}
To prove the result, it is enough to show that Assumptions (F) and (G) hold. Since $h(x)\sim x^{-r}$ for large $|x|$, we have
Since $\tilde{C}_1x^{-r}\leq h(x)\leq\tilde{C}_2 x^{-r}$ for sufficiently large $x$, for some constants $\tilde{C}_1,\tilde{C}_2>0$. Thus
$\int_{a_n}^{\infty}h(x/\lambda_n)/\lambda_n dx=\int_{a_n/\lambda_n}^{\infty}h(x) dx\leq\frac{\tilde{C}_2}{r-1}(a_n/\lambda_n)^{-(r-1)} $.
Given that $\lambda_n\leq a_np_n^{-(u'+1)/(r-1)}$ for some $u>0$,
$\frac{\tilde{C}_2}{r-1}(a_n/\lambda_n)^{-(r-1)}\leq \frac{\tilde{C}_2}{r-1}p_n^{-1-u'}\leq p_n^{-1-u}.$ Hence the (F) holds.
Also,
\begin{align*}
&-\log(\inf\limits_{x\in[-M_n,M_n]}h(x/\mu_n)/\mu_n)=-\log(\inf\limits_{x\in[-M_n/\mu_n,M_n/\mu_n]}h(x)/\mu_n)\leq -\log(\tilde{C}_1(M_n/\mu_n)^{-r}/\mu_n)\\
&=-\log(\tilde{C}_1)+r\log(M_n)+(r+1)\log(\mu_n)=O(\log(p_n)),
\end{align*}
verifying assumption (G).
\end{proof}

\bibliographystyle{natbib}
\bibliography{reference_biom}

\begin{thebibliography}{}

\bibitem[Ahfock {\em et~al.}(2017)Ahfock, Astle, and
  Richardson]{ahfock2017statistical}
Ahfock, D., Astle, W.~J., and Richardson, S. (2017).
\newblock Statistical properties of sketching algorithms.
\newblock {\em arXiv preprint arXiv:1706.03665\/}.

\bibitem[Ailon and Chazelle(2006)Ailon and Chazelle]{ailon2006approximate}
Ailon, N. and Chazelle, B. (2006).
\newblock Approximate nearest neighbors and the fast johnson-lindenstrauss
  transform.
\newblock In {\em Proceedings of the thirty-eighth annual ACM symposium on
  Theory of computing\/}, pages 557--563.

\bibitem[Ailon and Chazelle(2009)Ailon and Chazelle]{ailon2009fast}
Ailon, N. and Chazelle, B. (2009).
\newblock The fast johnson--lindenstrauss transform and approximate nearest
  neighbors.
\newblock {\em SIAM Journal on computing\/}, {\bf 39}(1), 302--322.

\bibitem[Armagan {\em et~al.}(2013)Armagan, Dunson, and
  Lee]{armagan2013generalized}
Armagan, A., Dunson, D.~B., and Lee, J. (2013).
\newblock Generalized double {P}areto shrinkage.
\newblock {\em Statistica Sinica\/}, {\bf 23}(1), 119--143.

\bibitem[Bhattacharya {\em et~al.}(2016)Bhattacharya, Chakraborty, and
  Mallick]{bhattacharya2016fast}
Bhattacharya, A., Chakraborty, A., and Mallick, B.~K. (2016).
\newblock Fast sampling with gaussian scale mixture priors in high-dimensional
  regression.
\newblock {\em Biometrika\/}, page asw042.

\bibitem[Candes and Tao(2006)Candes and Tao]{candes2006near}
Candes, E.~J. and Tao, T. (2006).
\newblock Near-optimal signal recovery from random projections: Universal
  encoding strategies?
\newblock {\em IEEE transactions on information theory\/}, {\bf 52}(12),
  5406--5425.

\bibitem[Caron and Doucet(2008)Caron and Doucet]{caron2008sparse}
Caron, F. and Doucet, A. (2008).
\newblock Sparse bayesian nonparametric regression.
\newblock In {\em Proceedings of the 25th international conference on Machine
  learning\/}, pages 88--95.

\bibitem[Carvalho {\em et~al.}(2010)Carvalho, Polson, and
  Scott]{carvalho2010horseshoe}
Carvalho, C.~M., Polson, N.~G., and Scott, J.~G. (2010).
\newblock The horseshoe estimator for sparse signals.
\newblock {\em Biometrika\/}, {\bf 97}(2), 465--480.

\bibitem[Castillo {\em et~al.}(2015)Castillo, Schmidt-Hieber, Van~der Vaart,
  {\em et~al.}]{castillo2015bayesian}
Castillo, I., Schmidt-Hieber, J., Van~der Vaart, A., {\em et~al.} (2015).
\newblock Bayesian linear regression with sparse priors.
\newblock {\em The Annals of Statistics\/}, {\bf 43}(5), 1986--2018.

\bibitem[Chen {\em et~al.}(2015)Chen, Liu, Lyu, King, and Zhang]{chen2015fast}
Chen, S., Liu, Y., Lyu, M.~R., King, I., and Zhang, S. (2015).
\newblock Fast relative-error approximation algorithm for ridge regression.
\newblock In {\em UAI\/}, pages 201--210.

\bibitem[Chowdhury {\em et~al.}(2018)Chowdhury, Yang, and
  Drineas]{chowdhury2018iterative}
Chowdhury, A., Yang, J., and Drineas, P. (2018).
\newblock An iterative, sketching-based framework for ridge regression.
\newblock In {\em International Conference on Machine Learning\/}, pages
  989--998.

\bibitem[Clarkson and Woodruff(2017)Clarkson and Woodruff]{clarkson2017low}
Clarkson, K.~L. and Woodruff, D.~P. (2017).
\newblock Low-rank approximation and regression in input sparsity time.
\newblock {\em Journal of the ACM (JACM)\/}, {\bf 63}(6), 1--45.

\bibitem[Dobriban and Liu(2018)Dobriban and Liu]{dobriban2018new}
Dobriban, E. and Liu, S. (2018).
\newblock A new theory for sketching in linear regression.
\newblock {\em arXiv preprint arXiv:1810.06089\/}.

\bibitem[Donoho(2006)Donoho]{donoho2006compressed}
Donoho, D.~L. (2006).
\newblock Compressed sensing.
\newblock {\em IEEE Transactions on information theory\/}, {\bf 52}(4),
  1289--1306.

\bibitem[Drineas {\em et~al.}(2011)Drineas, Mahoney, Muthukrishnan, and
  Sarl{\'o}s]{drineas2011faster}
Drineas, P., Mahoney, M.~W., Muthukrishnan, S., and Sarl{\'o}s, T. (2011).
\newblock Faster least squares approximation.
\newblock {\em Numerische mathematik\/}, {\bf 117}(2), 219--249.

\bibitem[Eldar and Kutyniok(2012)Eldar and Kutyniok]{eldar2012compressed}
Eldar, Y.~C. and Kutyniok, G. (2012).
\newblock {\em Compressed sensing: theory and applications\/}.
\newblock Cambridge university press.

\bibitem[George and McCulloch(1997)George and McCulloch]{george1997approaches}
George, E.~I. and McCulloch, R.~E. (1997).
\newblock Approaches for bayesian variable selection.
\newblock {\em Statistica sinica\/}, pages 339--373.

\bibitem[Ghosal {\em et~al.}(2007)Ghosal, Van Der~Vaart, {\em
  et~al.}]{ghosal2007convergence}
Ghosal, S., Van Der~Vaart, A., {\em et~al.} (2007).
\newblock Convergence rates of posterior distributions for noniid observations.
\newblock {\em Annals of Statistics\/}, {\bf 35}(1), 192--223.

\bibitem[Golub and Van~Loan(2012)Golub and Van~Loan]{golub2012matrix}
Golub, G.~H. and Van~Loan, C.~F. (2012).
\newblock {\em Matrix computations\/}, volume~3.
\newblock JHU press.

\bibitem[Griffin {\em et~al.}(2010)Griffin, Brown, {\em
  et~al.}]{griffin2010inference}
Griffin, J.~E., Brown, P.~J., {\em et~al.} (2010).
\newblock Inference with normal-gamma prior distributions in regression
  problems.
\newblock {\em Bayesian Analysis\/}, {\bf 5}(1), 171--188.

\bibitem[Guhaniyogi and Dunson(2015)Guhaniyogi and
  Dunson]{guhaniyogi2015bayesian}
Guhaniyogi, R. and Dunson, D.~B. (2015).
\newblock Bayesian compressed regression.
\newblock {\em Journal of the American Statistical Association\/}, {\bf
  110}(512), 1500--1514.

\bibitem[Guhaniyogi and Dunson(2016)Guhaniyogi and
  Dunson]{guhaniyogi2016compressed}
Guhaniyogi, R. and Dunson, D.~B. (2016).
\newblock Compressed gaussian process for manifold regression.
\newblock {\em The Journal of Machine Learning Research\/}, {\bf 17}(1),
  2472--2497.

\bibitem[Halko {\em et~al.}(2011)Halko, Martinsson, and
  Tropp]{halko2011finding}
Halko, N., Martinsson, P.-G., and Tropp, J.~A. (2011).
\newblock Finding structure with randomness: Probabilistic algorithms for
  constructing approximate matrix decompositions.
\newblock {\em SIAM review\/}, {\bf 53}(2), 217--288.

\bibitem[Huang(2018)Huang]{huang2018near}
Huang, Z. (2018).
\newblock Near optimal frequent directions for sketching dense and sparse
  matrices.
\newblock In {\em International Conference on Machine Learning\/}, pages
  2048--2057. PMLR.

\bibitem[Javanmard and Montanari(2014)Javanmard and
  Montanari]{javanmard2014confidence}
Javanmard, A. and Montanari, A. (2014).
\newblock Confidence intervals and hypothesis testing for high-dimensional
  regression.
\newblock {\em The Journal of Machine Learning Research\/}, {\bf 15}(1),
  2869--2909.

\bibitem[Jeong and Ghosal(2020)Jeong and Ghosal]{jeong2020unified}
Jeong, S. and Ghosal, S. (2020).
\newblock Unified bayesian asymptotic theory for sparse linear regression.
\newblock {\em arXiv preprint arXiv:2008.10230\/}.

\bibitem[Jiang(2007)Jiang]{jiang2007bayesian}
Jiang, W. (2007).
\newblock Bayesian variable selection for high dimensional generalized linear
  models: convergence rates of the fitted densities.
\newblock {\em The Annals of Statistics\/}, {\bf 35}(4), 1487--1511.

\bibitem[Johndrow {\em et~al.}(2020)Johndrow, Orenstein, and
  Bhattacharya]{johndrow2020scalable}
Johndrow, J.~E., Orenstein, P., and Bhattacharya, A. (2020).
\newblock Scalable approximate mcmc algorithms for the horseshoe prior.
\newblock {\em Journal of Machine Learning Research\/}, {\bf 21}(73), 1--61.

\bibitem[Mahoney(2011)Mahoney]{mahoney2011randomized}
Mahoney, M.~W. (2011).
\newblock Randomized algorithms for matrices and data.
\newblock {\em arXiv preprint arXiv:1104.5557\/}.

\bibitem[Maillard and Munos(2009)Maillard and Munos]{maillard2009compressed}
Maillard, O. and Munos, R. (2009).
\newblock Compressed least-squares regression.
\newblock {\em Advances in neural information processing systems\/}, {\bf 22},
  1213--1221.

\bibitem[Polson and Scott(2010)Polson and Scott]{polson2010shrink}
Polson, N.~G. and Scott, J.~G. (2010).
\newblock Shrink globally, act locally: Sparse bayesian regularization and
  prediction.
\newblock {\em Bayesian Statistics\/}, {\bf 9}, 501--538.

\bibitem[Raskutti and Mahoney(2016)Raskutti and
  Mahoney]{raskutti2016statistical}
Raskutti, G. and Mahoney, M.~W. (2016).
\newblock A statistical perspective on randomized sketching for ordinary
  least-squares.
\newblock {\em The Journal of Machine Learning Research\/}, {\bf 17}(1),
  7508--7538.

\bibitem[Rue(2001)Rue]{rue2001fast}
Rue, H. (2001).
\newblock Fast sampling of gaussian markov random fields.
\newblock {\em Journal of the Royal Statistical Society: Series B (Statistical
  Methodology)\/}, {\bf 63}(2), 325--338.

\bibitem[Sarlos(2006)Sarlos]{sarlos2006improved}
Sarlos, T. (2006).
\newblock Improved approximation algorithms for large matrices via random
  projections.
\newblock In {\em 2006 47th Annual IEEE Symposium on Foundations of Computer
  Science (FOCS'06)\/}, pages 143--152. IEEE.

\bibitem[Scott and Berger(2010)Scott and Berger]{scott2010bayes}
Scott, J.~G. and Berger, J.~O. (2010).
\newblock Bayes and empirical-bayes multiplicity adjustment in the
  variable-selection problem.
\newblock {\em The Annals of Statistics\/}, pages 2587--2619.

\bibitem[Song and Liang(2017)Song and Liang]{song2017nearly}
Song, Q. and Liang, F. (2017).
\newblock Nearly optimal bayesian shrinkage for high dimensional regression.
\newblock {\em arXiv preprint arXiv:1712.08964\/}.

\bibitem[Szab{\'o} {\em et~al.}(2015)Szab{\'o}, Van Der~Vaart, van Zanten, {\em
  et~al.}]{szabo2015frequentist}
Szab{\'o}, B., Van Der~Vaart, A.~W., van Zanten, J., {\em et~al.} (2015).
\newblock Frequentist coverage of adaptive nonparametric bayesian credible
  sets.
\newblock {\em The Annals of Statistics\/}, {\bf 43}(4), 1391--1428.

\bibitem[Ting {\em et~al.}(2008)Ting, Fienberg, and Trottini]{ting2008random}
Ting, D., Fienberg, S.~E., and Trottini, M. (2008).
\newblock Random orthogonal matrix masking methodology for microdata release.
\newblock {\em International Journal of Information and Computer Security\/},
  {\bf 2}(1), 86--105.

\bibitem[Vaart and Zanten(2011)Vaart and Zanten]{vaart2011information}
Vaart, A. v.~d. and Zanten, H.~v. (2011).
\newblock Information rates of nonparametric gaussian process methods.
\newblock {\em Journal of Machine Learning Research\/}, {\bf 12}(Jun),
  2095--2119.

\bibitem[Van~de Geer {\em et~al.}(2014)Van~de Geer, B{\"u}hlmann, Ritov,
  Dezeure, {\em et~al.}]{van2014asymptotically}
Van~de Geer, S., B{\"u}hlmann, P., Ritov, Y., Dezeure, R., {\em et~al.} (2014).
\newblock On asymptotically optimal confidence regions and tests for
  high-dimensional models.
\newblock {\em The Annals of Statistics\/}, {\bf 42}(3), 1166--1202.

\bibitem[Van Der~Pas {\em et~al.}(2014)Van Der~Pas, Kleijn, Van Der~Vaart, {\em
  et~al.}]{van2014horseshoe}
Van Der~Pas, S., Kleijn, B., Van Der~Vaart, A., {\em et~al.} (2014).
\newblock The horseshoe estimator: Posterior concentration around nearly black
  vectors.
\newblock {\em Electronic Journal of Statistics\/}, {\bf 8}(2), 2585--2618.

\bibitem[van~der Pas {\em et~al.}(2017)van~der Pas, Szab{\'o}, van~der Vaart,
  {\em et~al.}]{van2017adaptive}
van~der Pas, S., Szab{\'o}, B., van~der Vaart, A., {\em et~al.} (2017).
\newblock Adaptive posterior contraction rates for the horseshoe.
\newblock {\em Electronic Journal of Statistics\/}, {\bf 11}(2), 3196--3225.

\bibitem[Van~der Vaart {\em et~al.}(2009)Van~der Vaart, van Zanten, {\em
  et~al.}]{van2009adaptive}
Van~der Vaart, A.~W., van Zanten, J.~H., {\em et~al.} (2009).
\newblock Adaptive bayesian estimation using a gaussian random field with
  inverse gamma bandwidth.
\newblock {\em The Annals of Statistics\/}, {\bf 37}(5B), 2655--2675.

\bibitem[Vempala(2005)Vempala]{vempala2005random}
Vempala, S.~S. (2005).
\newblock {\em The random projection method\/}, volume~65.
\newblock American Mathematical Soc.

\bibitem[Vershynin(2010)Vershynin]{vershynin2010introduction}
Vershynin, R. (2010).
\newblock Introduction to the non-asymptotic analysis of random matrices.
\newblock {\em arXiv preprint arXiv:1011.3027\/}.

\bibitem[Wang {\em et~al.}(2017)Wang, Gittens, and Mahoney]{wang2017sketched}
Wang, S., Gittens, A., and Mahoney, M.~W. (2017).
\newblock Sketched ridge regression: Optimization perspective, statistical
  perspective, and model averaging.
\newblock {\em The Journal of Machine Learning Research\/}, {\bf 18}(1),
  8039--8088.

\bibitem[Woodruff(2014)Woodruff]{woodruff2014sketching}
Woodruff, D.~P. (2014).
\newblock Sketching as a tool for numerical linear algebra.
\newblock {\em arXiv preprint arXiv:1411.4357\/}.

\bibitem[Zhang {\em et~al.}(2010)Zhang {\em et~al.}]{zhang2010nearly}
Zhang, C.-H. {\em et~al.} (2010).
\newblock Nearly unbiased variable selection under minimax concave penalty.
\newblock {\em The Annals of statistics\/}, {\bf 38}(2), 894--942.

\bibitem[Zhang and Zhang(2014)Zhang and Zhang]{zhang2014confidence}
Zhang, C.-H. and Zhang, S.~S. (2014).
\newblock Confidence intervals for low dimensional parameters in high
  dimensional linear models.
\newblock {\em Journal of the Royal Statistical Society: Series B (Statistical
  Methodology)\/}, {\bf 76}(1), 217--242.

\bibitem[Zhang {\em et~al.}(2013)Zhang, Mahdavi, Jin, Yang, and
  Zhu]{zhang2013recovering}
Zhang, L., Mahdavi, M., Jin, R., Yang, T., and Zhu, S. (2013).
\newblock Recovering the optimal solution by dual random projection.
\newblock In {\em Conference on Learning Theory\/}, pages 135--157.

\bibitem[Zhao and Chen(2019)Zhao and Chen]{zhao2019privacy}
Zhao, L. and Chen, L. (2019).
\newblock On the privacy of matrix masking-based verifiable (outsourced)
  computation.
\newblock {\em IEEE Transactions on Cloud Computing\/}.

\bibitem[Zhou {\em et~al.}(2008)Zhou, Wasserman, and
  Lafferty]{zhou2008compressed}
Zhou, S., Wasserman, L., and Lafferty, J.~D. (2008).
\newblock Compressed regression.
\newblock In {\em Advances in Neural Information Processing Systems\/}, pages
  1713--1720.

\end{thebibliography}
\end{document}